\documentclass[10pt]{article}

\usepackage{authblk}
\usepackage{natbib}
\usepackage{mathtools}
\usepackage{amsmath}
\usepackage{amsthm}
\usepackage{amssymb}
\usepackage{amsbsy}
\usepackage{amsfonts}
\usepackage{amscd}
\usepackage{mathrsfs}
\usepackage{bm}
\usepackage{breqn}
\usepackage{color, soul}
\usepackage{enumerate}
\usepackage{empheq}
\usepackage{rotating}
\usepackage{ftnxtra}
\usepackage{fnpos}
\usepackage[titletoc,toc,title]{appendix}
\usepackage{euscript}
\usepackage{graphicx}
\usepackage{epsfig}
\usepackage{epstopdf}
\DeclareGraphicsExtensions{.pdf,.png,.jpg,.eps}
\usepackage{pstool}
\usepackage{upgreek}
\usepackage{mathrsfs}

\usepackage{psfrag}

\numberwithin{equation}{section}

\theoremstyle{plain}	
\newtheorem{thm}{Theorem}[section]

\newtheorem{prop}[thm]{Proposition}
\newtheorem*{prop*}{Proposition}
\theoremstyle{definition}	

\newtheorem{remark}[thm]{Remark}

\newtheorem{conj}[thm]{Conjecture} 

\usepackage{caption}
\usepackage{subcaption}

\usepackage{hyperref}
\hypersetup{colorlinks=true, linkcolor=blue}
\hypersetup{colorlinks=true,citecolor=blue}

\DeclareMathAlphabet{\mathpzc}{OT1}{pzc}{m}{it}

\usepackage{amsmath, amsthm, amssymb}

\usepackage{cleveref}

\DeclarePairedDelimiter\abs{\lvert}{\rvert}

\makeatletter
\newsavebox{\@brx}
\newcommand{\llangle}[1][]{\savebox{\@brx}{\(\m@th{#1\langle}\)}%
  \mathopen{\copy\@brx\mkern2mu\kern-0.9\wd\@brx\usebox{\@brx}}}
\newcommand{\rrangle}[1][]{\savebox{\@brx}{\(\m@th{#1\rangle}\)}%
  \mathclose{\copy\@brx\mkern2mu\kern-0.9\wd\@brx\usebox{\@brx}}}%
\let\oldabs\abs
\def\abs{\@ifstar{\oldabs}{\oldabs*}}
\makeatother

\usepackage{accents}
\newcommand{\cloak}[1]{\accentset{\Xi}{#1}}

    \newenvironment{mbmatrix}{\begin{bmatrix}}%
    {\end{bmatrix}}%

\begin{document}

\title{\textbf{Elastodynamic Transformation Cloaking for Non-Centrosymmetric Gradient Solids}}

\author[1]{Fabio Sozio}
\author[2]{Ashkan Golgoon}
\author[1,3]{Arash Yavari\thanks{Corresponding author, e-mail: arash.yavari@ce.gatech.edu}}
\affil[1]{\small \textit{School of Civil and Environmental Engineering, Georgia Institute of Technology, Atlanta, GA 30332, USA}}
\affil[2]{\small \textit{Department of Mechanical Engineering, Northwestern University, Evanston, IL 60208, USA}}
\affil[3]{\small \textit{The George W. Woodruff School of Mechanical Engineering, Georgia Institute of Technology, Atlanta, GA 30332, USA}}

\maketitle

\begin{abstract} 
In this paper we investigate the possibility of elastodynamic transformation cloaking in bodies made of non-centrosymmetric gradient solids. The goal of transformation cloaking is to hide a hole from elastic disturbances in the sense that the mechanical response of a homogeneous and isotropic body with a hole covered by a cloak would be identical to that of the corresponding homogeneous and isotropic body outside the cloak. It is known that in the case of centrosymmetric gradient solids the balance of angular momentum is the obstruction to transformation cloaking. We will show that this is the case for non-centrosymmetric gradient solids as well.
\end{abstract}

\begin{description}
\item[Keywords:] Cloaking, Gradient Elasticity, Cosserat Elasticity, Elastic Waves, Non-Centrosymmetric Solids, Chiral Solids.
\end{description}


\section{Introduction}

The idea of transformation cloaking in electromagnetism goes back to the works of \citet{Pendry2006} and \citet{Leonhardt2006}. Many researchers have tried to use the idea of transformation cloaking in other fields. In the case of elastodynamics, this has led to many inconsistent formulations that were critically reviewed in \citep{Yavari2019} and \citep{Golgoon2020}.
We should emphasize that the ideas related to elastodynamic cloaking are much older and go back to the works \citep{Gurney1938,Reissner1949,Mansfield1953} on reinforced holes in elastic sheets, and \citep{Hashin1962,Hashin1963,Hashin1985,Hashin1964,Benveniste2003} on neutral inhomogeneities.

Cloaking a hole in an elastic body can be formulated in terms of two equivalent boundary-value problems \citep{Yavari2019}. The hole is covered by a cloak that needs to be designed. The cloak is expected to have inhomogeneous mass density and inhomogeneous and anisotropic elastic properties. Outside the cloak the response of the body (the physical body $\mathcal{B}$) is required to be identical to that of a homogeneous and isotropic body with an infinitesimal hole (the virtual body $\tilde{\mathcal{B}}$). The two bodies are under the same external loads and have the same boundary conditions outside the cloak. In transformation cloaking one uses a map $\Xi:\mathcal{B}\rightarrow\tilde{\mathcal{B}}$ (cloaking map) that has two properties: i) Outside the cloak $\mathcal{C}$ it is the identity map, and ii) while fixing the outer boundary of the cloak it shrinks its inner boundary to a very small hole. Starting from the balance of linear momentum in one configuration one transforms it using the Piola transform to the other configuration. This gives transformation relations for the mass density and the elastic constants assuming that the displacement fields in the two configurations are equal at the corresponding points. In order to have identical mechanical responses outside the cloak, the cloaking map needs to fix the outer boundary of the cloak to the first order;  both $\Xi$, and its derivative map $T\Xi$ must be identity maps on the outer boundary of the cloak. The last thing to check is the balance of angular momentum. In the case of classical linear elasticity, generalized Cosserat elasticity, and centrosymmetric gradient elasticity, starting from a homogeneous and isotropic virtual body in which the balance of angular momentum is satisfied, it turns out that the balance of angular momentum cannot be satisfied in the physical problem unless the cloaking map is the identity map everywhere. In other words, the balance of angular momentum is the obstruction to exact transformation cloaking \citep{Yavari2019}. In the case of elastic plates a set of \emph{cloaking compatibility equations} obstruct transformation cloaking.   

There has been a misconception in the literature that an elastic cloak should be made of a Cosserat solid (see \citep{Yavari2019} for an extensive literature review). In \citep{Yavari2019} it was shown that even in the case of generalized Cosserat solids the balance of angular momentum is still the obstruction to transformation cloaking. No assumption was made on the elastic constants other than objectivity, and positive-definiteness of the elastic energy. This means that transformation cloaking is not possible in either non-centrosymmetric or centrosymmetric generalized Cosserat solids (and consequently Cosserat solids). \citet{Yavari2019} proved the impossibility of exact transformation cloaking for centrosymmetric gradient solids. In this paper we investigate the possibility of transformation cloaking for non-centrosymmetric gradient solids.

Noncentrosymmetric solids can be modeled in the setting of generalized continuum mechanics and have been studied by many researchers \citep{Cheverton1981,Lakes1982,Lakes2001,Sharma2004,liu2012chiral,iecsan2016chiral,Bohmer2020}.
\citet{Papanicolopulos2011} studied chirality in 3D isotropic gradient elasticity under the assumption of small strains. Chirality is controlled by a single material parameter in the fifth-order coupling elasticity tensor.
\citet{auffray2015complete,auffray2017handbook} studied the material symmetries in 2D linear gradient elasticity.
In dimension two, chirality is due to the lack of mirror symmetry, and it affects both the coupling and the second-order elasticity tensors. They showed that there are fourteen symmetry classes, eight of which have isotropic first-order elasticity tensors. 
In an effort to use chirality for cloaking applications, \citet{Nassar2019} considered a sheet made of a classical linear elastic solid connected to an elastic foundation that resists rotations. They called such structures ``polar solids", which is a misleading term; the energy functions they considered are not objective. Also, their cloaking structure construction cannot be generalized to 3D.

This paper is structured as follows. In \S2 we tersely review the governing equations of elastodynamics. In \S3 gradient elasticity, its governing equations, and non-centrosymmetry are discussed. We formulate the problem of transformation cloaking in linearized non-centrosymmetric gradient elasticity in \S4. We prove the impossibility of cloaking for arbitrary cylindrical holes and arbitrary cloaking maps. 
Conclusions are given in \S5.

\section{Nonlinear Elasticity}

\paragraph{Kinematics.}
In nonlinear elasticity, motion is a time-dependent mapping between a reference configuration (or natural configuration) and the ambient space. We write this as $\varphi_t:\mathcal{B}\rightarrow\mathcal{S}$, where $(\mathcal{B},\mathbf{G})$ and $(\mathcal{S},\mathbf{g})$ are the material and the ambient space Riemannian manifolds, respectively \citep{MaHu1983}. Here, $\mathbf{G}$ is the material metric (that allows one to measure distances in a natural stress-free configuration) and $\mathbf{g}$ is the background metric of the ambient space. The Levi-Civita connections associated with the metrics $\mathbf{G}$ and $\mathbf{g}$ are denoted as $\nabla^{\mathbf{G}}$ and $\nabla^{\mathbf{g}}$, respectively. The corresponding Christoffel symbols of $\nabla^{\mathbf{G}}$ and $\nabla^{\mathbf{g}}$ in the local coordinate charts $\{X^A\}$ and $\{x^a\}$ are denoted by $\Gamma^A{}_{BC}$ and $\gamma^a{}_{bc}$, respectively. These can be directly expressed in terms of the metric components as
\begin{equation}
	\gamma^a{}_{bc}=\frac{1}{2}g^{ak}\left(g_{kb,c}+g_{kc,b}-g_{bc,k}\right)\,,\qquad \Gamma^A{}_{BC}
	=\frac{1}{2}G^{AK}\left(G_{KB,C}+G_{KC,B}-G_{BC,K}\right)\,.
\end{equation}
The deformation gradient $\mathbf{F}$ is the tangent map of $\varphi_t$, which is defined as $\mathbf{F}(X,t)=T\varphi_t(X): T_X\mathcal{B}\rightarrow T_{\varphi_t(X)}\mathcal{S}$. The transpose of $\mathbf{F}$ is denoted by $\mathbf{F}^{\mathsf{T}}$, where
\begin{equation}
	\mathbf{F}^{\mathsf{T}}(X,t):T_{\varphi_t(X)}\mathcal{S}\rightarrow T_X\mathcal{B}\,,\qquad 
	\llangle \mathbf{W},\mathbf{F}^{\mathsf{T}}\mathbf{w} \rrangle_{\mathbf{G}}
	=\llangle \mathbf F \mathbf W,\mathbf{w} \rrangle_{\mathbf{g}},
	~~ \forall\, \mathbf{W}\in T_X\mathcal{B},\, \mathbf{w}\in T_{\varphi_t(X)}\mathcal{S}\,.
\end{equation}
In components, $(F^{\mathsf{T}})^A{}_a=G^{AB}F^b{}_Bg_{ab}$.
The right Cauchy-Green deformation tensor is defined as $\mathbf{C}=\mathbf{F}^{\mathsf T}\mathbf{F}:T_X\mathcal{B}\rightarrow T_X\mathcal{B}$, which in components reads $C^A{}_B=F^a{}_LF^b{}_Bg_{ab}G^{AL}$. Note that $\mathbf{C}^\flat=\varphi_t^*\mathbf{g}$.

The material velocity of the motion is the mapping $\mathbf{V}:\mathcal{B}\times\mathbb{R}^+\rightarrow T\mathcal{S}$, where $\mathbf{V}(X,t)\in T_{\varphi_t(X)}\mathcal{S}$, and in components, $V^a(X,t)=\frac{\partial \varphi^a}{\partial t}(X,t)$. 
The material acceleration is a mapping $\mathbf{A}:\mathcal{B}\times\mathbb{R}^+\rightarrow T\mathcal{S}$ defined as $\mathbf{A}(X,t):=D^{\mathbf{g}}_{t}\mathbf{V}(X,t)=\nabla^{\mathbf{g}}_{\mathbf{V}(X,t)}\mathbf{V}(X,t)\in T_{\varphi_t(X)}\mathcal{S}$, where $D^{\mathbf{g}}_{t}$ denotes the covariant derivative along the curve $\varphi_t(X)$ in $\mathcal{S}$. In components, $A^a=\frac{\partial V^a}{\partial t}+\gamma^a{}_{bc}V^bV^c$.

\paragraph{Balance laws.}
The balance of linear momentum in material form reads
\begin{equation}\label{linear-momentum}
    \operatorname{Div}\mathbf{P}+\rho_0\mathbf{B}=\rho_0\mathbf{A},
\end{equation}
where $\mathbf{P}$ is the first Piola-Kirchhoff stress. $\rho_0$, $\mathbf{B}$, and $\mathbf{A}$ are the material mass density, material body force, and material acceleration, respectively. $\operatorname{Div}\mathbf{P}$ has the following coordinate expression
\begin{equation}
	\operatorname{Div}\mathbf{P}=P^{aA}{}_{|A}\frac{\partial}{\partial x^a}
	=\left(\frac{\partial P^{aA}}{\partial X^A}+P^{aB}\Gamma^A{}_{AB} 
	+P^{cA}F^b{}_A\gamma^a{}_{bc}\right)\frac{\partial}{\partial x^a}\,.
\end{equation}
The Jacobian of deformation $J$ relates the deformed and undeformed Riemannian volume elements as $dv(x,\mathbf{g})=JdV(X,\mathbf{G})$, and is written as
\begin{equation}
	J=\sqrt{\frac{\det\mathbf{g}}{\det\mathbf{G}}}\det\mathbf{F}.
\end{equation}
Identifying a material point with its position in the material manifold $X\in\mathcal{B}$, we have $x=\varphi_t(X)$. When the ambient space is Euclidean one defines the material displacement field as $\mathbf{U}=\varphi_t(X)-X$.\footnote{In \S3.1, in linearized gradient elasticity we will use $\mathbf{U}$ for the linearized displacement instead of $\delta\mathbf{U}$.}

Balance of angular momentum in local form reads $\mathbf{F}\mathbf{P}^{\star}=\mathbf{P}\mathbf{F}^{\star}$, where $\mathbf{P}^{\star}$ and $\mathbf{F}^{\star}$ are duals of $\mathbf{P}$ and $\mathbf{F}$, respectively, and are defined as
\begin{equation}
\begin{aligned}
	& \mathbf{F}=F^a{}_A \frac{\partial}{\partial x^a}\otimes dX^A,~~~
	&&\mathbf{F}^{\star}=F^a{}_A dX^A \otimes \frac{\partial}{\partial x^a}\,,\\
	& \mathbf{P}=P^{aA} \frac{\partial}{\partial x^a}\otimes \frac{\partial}{\partial X^A},~~~
	&&\mathbf{P}^{\star}=P^{aA} \frac{\partial}{\partial X^A}\otimes \frac{\partial}{\partial x^a}\,.
\end{aligned}
\end{equation}
Note that $\mathbf{F}^\star:T_{\varphi_t(X)}^*\mathcal{S}\rightarrow T_X^*\mathcal{B}$, where $T_{\varphi_t(X)}^*\mathcal{S}$ and $T_X^*\mathcal{B}$ denote the cotangent spaces of $T_{\varphi_t(X)}\mathcal{S}$ and $T_X\mathcal{B}$, respectively. Balance of angular momentum in components reads $F^a{}_AP^{bA}=F^b{}_AP^{aA}$.

Conservation of mass implies that $\rho dv=\rho_0 dV$ or $\rho J=\rho_0$, where $\rho_o$ and $\rho$ denote the material and spatial mass densities, respectively. In terms of Lie derivatives, conservation of mass can be written as $\mathbf{L}_{\mathbf{v}}\rho=0$ \citep{MaHu1983}.

\section{Gradient Elasticity}

In this section we extend the analysis of \citet{Yavari2019} to non-centrosymmetric solids. We refer the reader to \citep{Yavari2019} for the detailed derivation of the governing equations and the transformed fields.
In gradient elasticity (or strain-gradient elasticity) energy function has the following form \citep{Toupin1964}
\begin{equation}
	W=W(X,\mathbf{F},\nabla\mathbf{F},\mathbf{G},\mathbf{g}\circ\varphi)\,.
\end{equation}
From compatibility equations $F^a{}_{A|B}=F^a{}_{B|A}$ \citep{Yavari2013}.
Material frame indifference (objectivity) implies that \citep{Toupin1964, Yavari2019} $W=\hat{W}(X,C_{AB},C_{AB|C},G_{AB})$.
The first Piola-Kirchhoff stress has the following representation
\begin{equation}\label{Piola-Gradient}
	P^{aA}=g^{ab}\left[\frac{\partial W}{\partial F^b{}_A}
	-\left(\frac{\partial W}{\partial F^b{}_{A|B}}\right)_{|B} \right].
\end{equation}
Hyper-stress is defined as $H_a{}^{AB}=H_a{}^{BA}=\frac{\partial W}{\partial F^a{}_{A|B}}$. 
Traction is written as
\begin{equation} \label{Traction-Toupin}
	T^a=P^{aA}N_A-H^{aAB}{}_{|B}N_A+H^{aAB}\mathfrak{B}_{AB},
\end{equation}
where $\mathfrak{B}_{AB}=\mathfrak{B}_{BA}=-N_{A|B}$ is the second fundamental form of the surface $\partial\mathcal{B}$ embedded in the Euclidean space, and $\mathbf{N}$ is the unit normal vector to $\partial\mathcal{B}$. 
Note that in a stress-free gradient solid both the (total) first Piola-Kirchhoff stress and hyper-stress vanish. 

\subsection{Balance of linear and angular momenta}

In terms of the first Piola-Kirchhoff stress the balance of angular momentum reads
\begin{equation}
	P^{[aA}F^{b]}{}_A+\left(H^{[aAB}F^{b]}{}_A\right)_{|B}=0\,.
\end{equation}
Linearizing the balance of linear momentum about a motion $\mathring{\varphi}$ one obtains $(\delta P^{aA})_{|A}+\rho_0\delta B^a=\rho_0\ddot{U}^a$, where
\begin{equation}
	\delta P^{aA}=\frac{\partial P^{aA}}{\partial F^b{}_B}\delta F^b{}_B
	+\frac{\partial P^{aA}}{\partial F^b{}_{B|C}}\delta F^b{}_{B|C}
	=\mathsf{A}^{aA}{}_b{}^B~U^b{}_{|B}+\mathsf{B}^{aA}{}_b{}^{BC}~U^b{}_{|B|C}	\,,
\end{equation}
where $\delta$ denotes the first variation of a field, and $U^a$ are the components of the linearized displacement field, i.e., $U^a = \delta\varphi^a$, and
\begin{equation}
	\mathsf{A}^{aA}{}_b{}^B=\frac{\partial P^{aA}}{\partial F^b{}_B},~~~
	\mathsf{B}^{aA}{}_b{}^{BC}=\frac{\partial P^{aA}}{\partial F^b{}_{B|C}}\,.
\end{equation}
$\boldsymbol{\mathsf{A}}$ and $\boldsymbol{\mathsf{B}}$ are the dynamic elastic constants \citep{DiVincenzo1986}. Notice that $\mathsf{B}^{aAbBC}=\mathsf{B}^{aAbCB}$.
From $P^{aA}=g^{am}\frac{\partial W}{\partial F^m{}_A}-H^{aAM}{}_{|M}$, one writes
\begin{equation}
\begin{aligned}
	\delta P^{aA} &=g^{am}\frac{\partial^2 W}{\partial F^m{}_A \partial F^n{}_N} \delta F^n{}_N
	+g^{am}\frac{\partial^2 W}{\partial F^m{}_A \partial F^n{}_{N|M}}\delta F^n{}_{N|M}
	-\delta(H^{aAM}{}_{|M})\\
	&=g^{am}\frac{\partial^2 W}{\partial F^m{}_A \partial F^n{}_N}U^n{}_{|N}
	+g^{am}\frac{\partial^2 W}{\partial F^m{}_A \partial F^n{}_{N|M}}U^n{}_{|N|M}
	-(\delta H^{aAM})_{|M}
	\,.
\end{aligned}
\end{equation}
But
\begin{equation} \label{DeltaH}
\begin{aligned}
	\delta H^{aAM} &=\frac{\partial H^{aAM}}{\partial F^c{}_C} \delta F^c{}_C
	+\frac{\partial H^{aAM}}{\partial F^c{}_{C|D}} \delta F^c{}_{C|D}\\
	&=g^{am}\frac{\partial^2 W}{\partial F^c{}_C \partial F^m{}_{A|M}}U^c{}_{|C}
	+g^{am}\frac{\partial^2 W}{\partial F^c{}_{C|D} \partial F^m{}_{A|M}}U^c{}_{|C|D}
	\,.
\end{aligned}
\end{equation}
The static elastic constants are defined as \citep{DiVincenzo1986}
\begin{equation}\label{gr-elas}
	\mathbb{A}_a{}^A{}_b{}^B=\frac{\partial^2 W}{\partial F^a{}_A \partial F^b{}_B},~~~
	\mathbb{B}_a{}^A{}_b{}^{BC}=\frac{\partial^2 W}{\partial F^a{}_{A} \partial F^b{}_{B|C}},~~~
	\mathbb{C}_a{}^{AB}{}_b{}^{CD}=\frac{\partial^2 W}{\partial F^a{}_{A|B} \partial F^b{}_{C|D}}
	\,.
\end{equation}
The static elastic constants have the following symmetries: 
\begin{equation} \label{Symmetries}
\begin{aligned}
	& \mathbb{A}_a{}^A{}_b{}^B=\mathbb{A}_b{}^B{}_a{}^A, \\
	& \mathbb{B}_a{}^A{}_b{}^{BC}=\mathbb{B}_a{}^A{}_b{}^{CB}, \\
	& \mathbb{C}_a{}^{AB}{}_b{}^{CD}=\mathbb{C}_a{}^{BA}{}_b{}^{CD}
	=\mathbb{C}_a{}^{BA}{}_b{}^{DC}=\mathbb{C}_b{}^{DC}{}_a{}^{BA}\,.
\end{aligned}
\end{equation}
Thus, from \eqref{DeltaH} $\delta H^{aAM}=\mathbb{B}_b{}^{BaAM} U^b{}_{|B}+\mathbb{C}_b{}^{BCaAM}U^b{}_{|B|C}$, and hence
\begin{equation}
	(\delta H^{aAM})_{|M}=(\mathbb{B}_b{}^{BaAM} U^b{}_{|B}+\mathbb{C}_b{}^{BCaAM}U^b{}_{|B|C})_{|M}
	\,.
\end{equation}
Therefore
\begin{equation}\label{Elastic-Constants-Gradient}
\begin{aligned}
	\mathsf{A}^{aA}{}_b{}^B &=\mathbb{A}^{aA}{}_b{}^B-\mathbb{B}_b{}^{BaAM}{}_{|M},\\
	\mathsf{B}^{aA}{}_b{}^{BC} &= \mathbb{B}^{aA}{}_b{}^{BC}-\mathbb{B}_b{}^{BaAC}-\mathbb{C}^{aAM}{}_b{}^{BC}{}_{|M}
	\,.
\end{aligned}
\end{equation}
Or equivalently
\begin{equation} \label{Dynamics-Static}
\begin{aligned}
	\mathsf{A}^{aAbB} &=\mathbb{A}^{aAbB}-\mathbb{B}^{bBaAM}{}_{|M},\\
	\mathsf{B}^{aAbBC} &= \mathbb{B}^{aAbBC}-\mathbb{B}^{bBaAC}-\mathbb{C}^{aAMbBC}{}_{|M}
	\,.
\end{aligned}
\end{equation}
In deriving the second relation we ignored the term $U^b{}_{|B|C|M}$ in $\delta H^{aAM}{}_{|M}$ as we are assuming a second-gradient elasticity; displacement derivatives of orders three or higher are neglected. 

When linearized with respect to a stress-free initial configuration, the linearized balance of angular momentum is written as
\begin{equation}\label{Angular-Momentum-Gradient-Linear}
\begin{aligned}
	\mathbb{A}^{[aM}{}_m{}^A\mathring{F}^{b]}{}_M&=0,\\
	\mathbb{B}^{[aM}{}_m{}^{AB}\mathring{F}^{b]}{}_M &=0,
\end{aligned}
\end{equation}
and with an abuse of notation
\begin{equation}\label{Angular-Momentum-Gradient-Linear1}
	\mathbb{A}^{[ab]}{}_m{}^A =0,~~~\mathbb{B}^{[ab]}{}_m{}^{AB}=0.
\end{equation}

\subsection{The coupling elastic constants for isotropic solids}

Materials with non-vanishing coupling elasticity tensors $\mathbb{B}$ are those that are not invariant under inversions. These are called non-centrosymmetric solids.
These materials can be either chiral if they are not invariant under orientation-reversing transformations, or achiral.
According to \citet{Auffray2019}, the symmetry groups for these materials are of Type I (chiral) or of Type III (neither chiral nor centrosymmetric).
A further classification can be done using the property of polarity, i.e., the property of having a single rotational axis of symmetry. Therefore, in summary, non-centrosymmetric materials are divided into four cases: chiral polar, achiral polar, chiral apolar, achiral apolar.
Isotropic non-centrosymmetric solids are the isotropic chiral ones.

For non-centrosymmetric solids the coupling elastic constants do not vanish.
Let us consider the corresponding fifth-order elastic constants in terms of the right Cauchy-Green strain $\mathbf{C}$, namely
\begin{equation}
	\mathbb{L}^{ABCDE}=\frac{\partial^2 W}{\partial C_{AB} \partial C_{CD|E} }.
\end{equation}
$\boldsymbol{\mathbb{L}}$ has the minor $\mathbb{L}^{ABCDE}=\mathbb{L}^{BACDE}=\mathbb{L}^{ABDCE}$, and major $\mathbb{L}^{ABCDE}=\mathbb{L}^{CDEAB}$  symmetries. When the elastic constants are defined with respect to a stress-free initial configuration, one can show that
\begin{equation} \label{L-B-Relation}
	\mathbb{B}^{aAbBC}=
	4 \mathring{F}^{a}{}_{M}\mathring{F}^{b}{}_{N}~ \mathbb{L}^{AMNBC}.
\end{equation}
For an isotropic solid, in Cartesian coordinates one has the following representation for $\boldsymbol{\mathbb{L}}$ \citep{Suiker2000}:
\begin{equation}
\begin{aligned}
	\mathbb{L}^{IJKLM}
	= & \ell_1 \epsilon^{IJK}\delta^{LM} 
	+ \ell_2 \epsilon^{IJL}\delta^{KM}
	+ \ell_3 \epsilon^{IJM}\delta^{KL}
	+ \ell_4 \epsilon^{IKL}\delta^{JM}
	+ \ell_5 \epsilon^{IKM}\delta^{JL} \\
	&+ \ell_6 \epsilon^{ILM}\delta^{JK}
	+ \ell_7 \epsilon^{JKL}\delta^{IM}
	+ \ell_8 \epsilon^{JKM}\delta^{IL}
	+ \ell_9 \epsilon^{JLM}\delta^{IK}
	+ \ell_{10} \epsilon^{KLM}\delta^{IJ},
\end{aligned}
\end{equation}
where $\epsilon^{IJK}$ is the permutation symbol, and $\ell_i$ are elastic constants. From $C_{IJ}=C_{JI}$ we have the minor symmetries $\mathbb{L}^{IJKLM}=\mathbb{L}^{JIKLM}$, and $\mathbb{L}^{IJKLM}=\mathbb{L}^{JILKM}$, which dictate $\ell_1=\ell_2=\ell_3=\ell_4=\ell_7=\ell_{10}=0$.\footnote{%
Once the balance of angular momentum is enforced, both coupling elasticity tensors $\mathbb{B}$ and $\mathbb{L}$ have $108$ independent components in the most general case. In \citep{Yavari2019} it was mentioned that $\mathbb{B}$ has $90$ independent components, which is incorrect. However, this inaccurate statement did not affect any of the results or conclusions of that work.}
Thus
\begin{equation}
	\mathbb{L}^{IJKLM} = 
	\left(\ell_5 \epsilon^{IKM}\delta^{JL} 
	+ \ell_8 \epsilon^{JKM}\delta^{IL}\right)
	+ \left(\ell_6 \epsilon^{ILM}\delta^{JK}
	+ \ell_9 \epsilon^{JLM}\delta^{IK}\right).
\end{equation}
Looking at the contribution of $\boldsymbol{\mathbb{L}}$ to energy, one can see that due to symmetry of the right Cauchy-Green strain only the sum of the remaining four elastic constants $(\ell_5+\ell_8+\ell_6+\ell_9)$ appears in the energy expression. This implies that there is only one elastic constant $b_0$, and
\begin{equation} \label{L-Tensor}
	\mathbb{L}^{IJKLM} = 
	b_0\left(\epsilon^{IKM}\delta^{JL}
	+\epsilon^{JKM}\delta^{IL}
	+\epsilon^{ILM}\delta^{JK}
	+\epsilon^{JLM}\delta^{IK}\right).
\end{equation}
This is consistent with the results of \citet{Dell2009}, \citet{Papanicolopulos2011}, and \citet{Auffray2019}.\footnote{Note that this tensor does not have any major symmetries; the symmetries claimed in Eq.(3.2)$_2$ in \citep{Dell2009} are incorrect. As a matter of fact, from the representation \eqref{L-Tensor} in the isotropic case the coupling elasticity tensor $\mathbb{L}$ has the following major antisymmetry: $\mathbb{L}^{IJKLM} = -\mathbb{L}^{KLIJM}$. From \eqref{L-B-Relation}, $\mathbb{B}$ has the same property in the isotropic case.}
In arbitrary curvilinear coordinates \eqref{L-Tensor} is written as
\begin{equation} 
	\mathbb{L}^{IJKLM} = 
	b_0\left(\varepsilon^{IKM}g^{JL}
	+\varepsilon^{JKM}g^{IL}
	+\varepsilon^{ILM}g^{JK}
	+\varepsilon^{JLM}g^{IK}\right),
\end{equation}
where $\varepsilon^{IJK}=\frac{1}{\sqrt{g}}\epsilon^{IJK}$, and $g=\operatorname{det}\mathbf{g}$.

\subsection{Positive-definiteness of the elastic energy}

Starting from a stress-free initial configuration, the change in the elastic energy is written as
\begin{equation}
\begin{aligned}
\delta W &=\frac{1}{2}\frac{\partial^2 W}{\partial F^a{}_A\partial F^b{}_B}U^a{}_{|A}U^b{}_{|B}
+\frac{\partial^2 W}{\partial F^a{}_A\partial F^b{}_{B|C}}U^a{}_{|A}U^b{}_{|B|C}
+\frac{1}{2}\frac{\partial^2 W}{\partial F^a{}_{A|B}\partial F^b{}_{C|D}}U^a{}_{|A|B}U^b{}_{|C|D} \\
& = \frac{1}{2}\mathbb{A}^{aAbB}U_{a|A}U_{b|B}
+\mathbb{B}^{aAbBC}U_{a|A}U_{b|B|C}
+\frac{1}{2}\mathbb{C}^{aABbCD}U_{a|A|B}U_{b|C|D}.
\end{aligned}
\end{equation}
Positive-definiteness of the elastic energy requires that $\delta W>0$ for any pair $(U_{a|A},U_{a|A|B})\neq(0,0)$. In particular, when $U_{a|A}\neq 0$, and $U_{a|A|B}=0$, $\mathbb{A}^{aAbB}U_{a|A}U_{b|B}>0$, which implies that $\boldsymbol{\mathbb{A}}$ must be positive-definite. In the case of isotropic solids this is equivalent to $\mu>0$, and $3\lambda+2\mu>0$. Similarly, $\boldsymbol{\mathbb{C}}$ must be positive-definite. It turns out that $-k<b_0<k$, where $k$ depends on $\mu$ and two sixth-order elastic constants \citep{Papanicolopulos2011}.

\section{Transformation Cloaking in Linearized Gradient Elastodynamics}

\subsection{Shifters in Euclidean ambient space}

It is assumed that the reference configurations of both the physical and virtual bodies are embedded in the Euclidean space. In order to relate vector fields in the physical problem to those in the virtual problem one uses shifters. 
We assume that $\mathcal{B}\subset\mathcal{S}=\mathbb{R}^n$ ($n=2$ or $3$). 
The shifter map is the map $\boldsymbol{\mathsf{s}}:T\mathcal{S}\rightarrow T\mathcal{S}$, $\boldsymbol{\mathsf{s}}(x,\mathbf{w})=(\tilde{x},\mathbf{w})$. 
Its restriction to $x\in\mathcal{S}$ is denoted by $\boldsymbol{\mathsf{s}}_x=\boldsymbol{\mathsf{s}}(x):T_x\mathcal{S}\rightarrow T_{\tilde{x}}\mathcal{S}$, and parallel transports $\mathbf{w}$ based at $x\in\mathcal{S}$ to $\mathbf{w}$ based at $\tilde{x}\in\mathcal{S}$ (see Fig. \ref{fig:shifters}). 
Let us choose two global colinear Cartesian coordinates $\{\tilde{z}^{\tilde{i}}\}$ and $\{z^i\}$ for the virtual and physical deformed configurations in the ambient space. Also consider curvilinear coordinates $\{\tilde{x}^{\tilde{a}}\}$ and $\{x^a\}$ for the two configurations.
Noting that $\mathsf{s}^{\tilde{i}}{}_i=\delta^{\tilde{i}}_i$, one can show that \citep{MaHu1983}
\begin{equation}
	\mathsf{s}^{\tilde{a}}{}_a(x)=\frac{\partial \tilde{x}^{\tilde{a}}}{\partial \tilde{z}^{\tilde{i}}}(\tilde{x})
	\frac{\partial z^i}{\partial x^a}(x)	\delta^{\tilde{i}}_i \,.
\end{equation}
As an example, in the cylindrical coordinates $(r,\theta,z)$ and $(\tilde{r},\tilde{\theta},\tilde{z})$ at $x\in\mathbb{R}^3$ and $\tilde{x}\in\mathbb{R}^3$, respectively, one can show that the shifter map has the following matrix representation
\begin{equation} \label{s-cylindrical}
	\boldsymbol{\mathsf{s}}=
\begin{bmatrix} 
  \cos(\tilde{\theta}-\theta) &  r\sin(\tilde{\theta}-\theta) & 0  \\
  -\sin(\tilde{\theta}-\theta)/\tilde{r} & r\cos(\tilde{\theta}-\theta)/\tilde{r}  & 0  \\
  0 & 0  &   1
\end{bmatrix} \,.
\end{equation}
Similarly, in the spherical coordinates $(r,\theta,\phi)$ and $(\tilde{r},\tilde{\theta},\tilde{\phi})$ at $x\in\mathbb{R}^3$ and $\tilde{x}\in\mathbb{R}^3$, respectively, the shifter map has the following matrix representation
\begin{equation} \label{s-spherical}
	\boldsymbol{\mathsf{s}}=
\begin{mbmatrix} 
  \cos(\tilde{\phi}-\phi)\sin\tilde{\theta}\sin\theta+\cos\tilde{\theta}\cos\theta &  
  r[\cos(\tilde{\phi}-\phi)\sin\tilde{\theta}\cos\theta-\cos\tilde{\theta}\sin\theta] 
  & r\sin(\tilde{\phi}-\phi)\sin\tilde{\theta}\sin\theta  \\
  [\cos(\tilde{\phi}-\phi)\cos\tilde{\theta}\sin\theta-\sin\tilde{\theta}\cos\theta]/\tilde{r} & 
  r[\cos(\tilde{\phi}-\phi)\cos\tilde{\theta}\cos\theta+\sin\tilde{\theta}\sin\theta]/\tilde{r}  & 
  r\sin(\tilde{\phi}-\phi)\cos\tilde{\theta}\sin\theta/\tilde{r}  \\
  -\sin(\tilde{\phi}-\phi)\sin\theta/(\tilde{r}\sin\tilde{\theta})  & 
  -r\sin(\tilde{\phi}-\phi)\cos\theta/(\tilde{r}\sin\tilde{\theta})
  &  r\cos(\tilde{\phi}-\phi)\sin\theta/(\tilde{r}\sin\tilde{\theta})
\end{mbmatrix} \,.
\end{equation}
\begin{figure}[t]
\centering
\vskip 0.1in
\includegraphics[width=.55\textwidth]{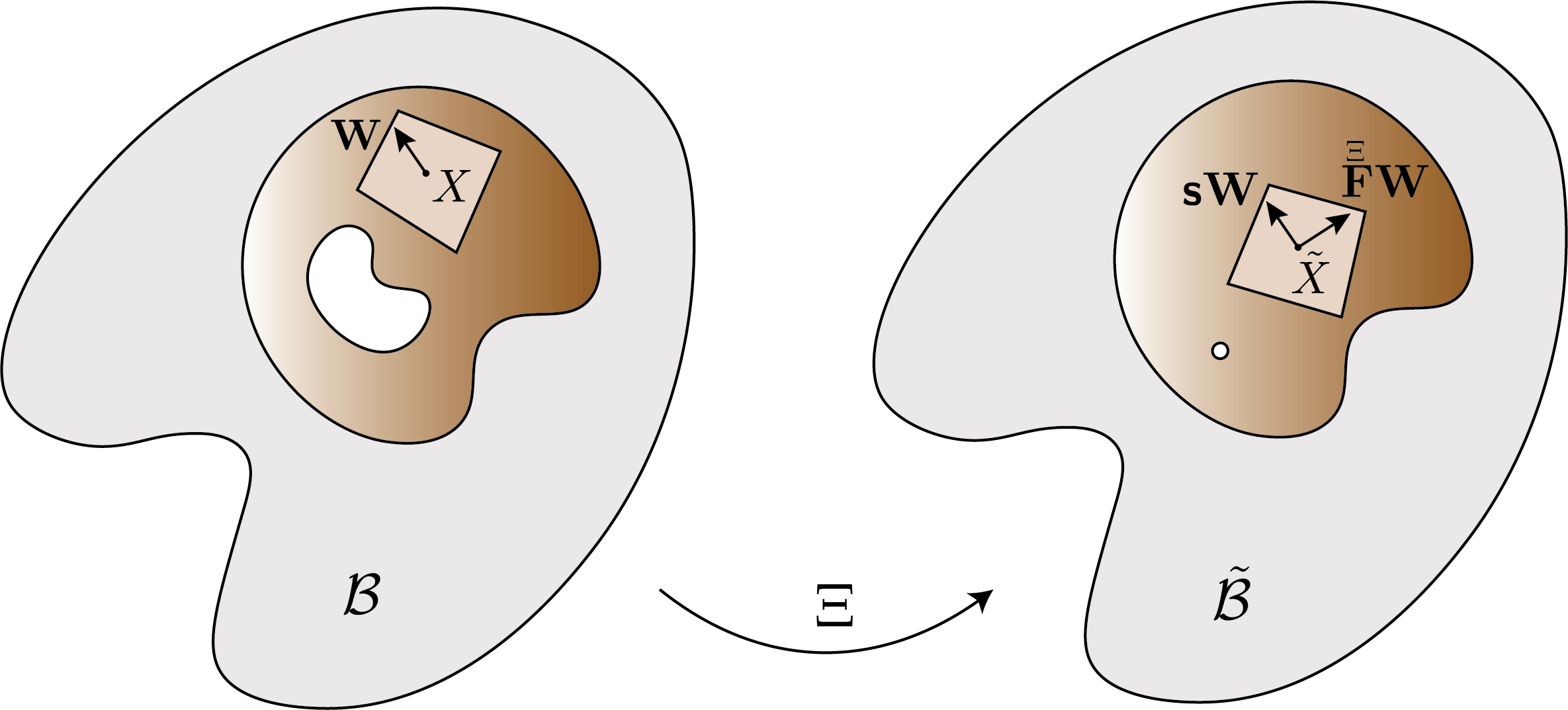}
\vskip 0.1in
\caption{$\Xi:\mathcal{B}\rightarrow\tilde{\mathcal{B}}$ is a map between two submanifolds of $\mathbb{R}^n$. The shifter map $\boldsymbol{\mathsf{s}}$ parallel transports $\mathbf{W}$ at $X$ to $\boldsymbol{\mathsf{s}}\mathbf{W}$ at $\tilde{X}=\Xi(X)$.}
\label{fig:shifters}
\end{figure}
In Fig. \ref{fig:radial-shifters} a radial map $\Xi:\mathcal{B}\rightarrow\tilde{\mathcal{B}}$ is shown. The shifter parallel transports a vector $\mathbf{W}$ at $X=(R,\Theta,Z)$ (or $X=(R,\Theta,\Phi)$) to $(f(R),\Theta,Z)$ (or $(f(R),\Theta,\Phi)$).

\begin{figure}[hbt]
\centering
\vskip 0.2in
\includegraphics[width=.55\textwidth]{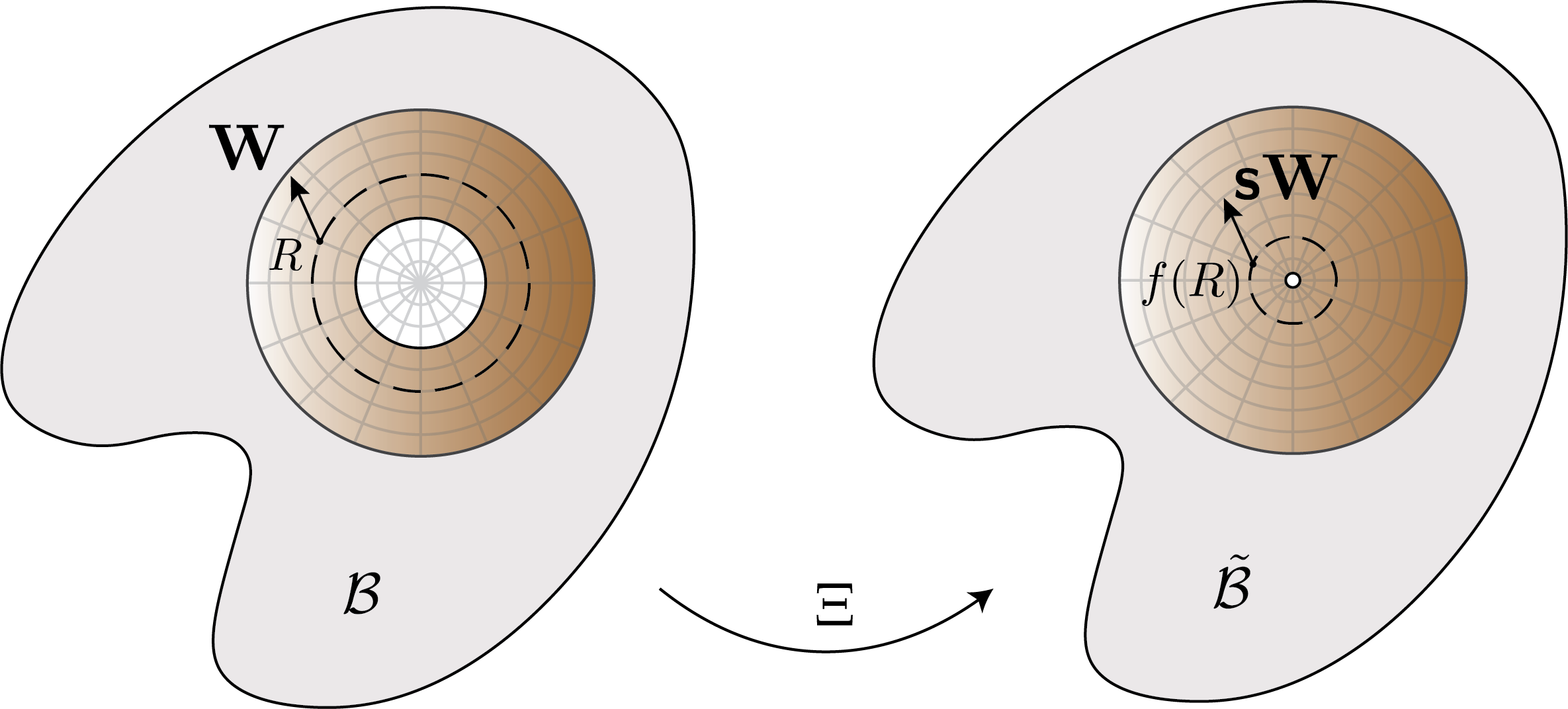}
\vskip 0.1in
\caption{Shifters along a cylindrically/spherically-symmetric map $\Xi:(R,\Theta,Z)\mapsto(f(R),\Theta,Z)$ or $\Xi:(R,\Theta,\Phi)\mapsto(f(R),\Theta,\Phi)$.
The shifter map $\boldsymbol{\mathsf{s}}$ parallel transports $\mathbf{W}$ in the radial direction from $R$ to $\tilde R = f(R)$.}
\label{fig:radial-shifters}
\end{figure}



\subsection{Transformation cloaking formulated as equivalent boundary-value problems}

In the coordinate charts $\{X^A\}$ and $\{x^a\}$, the divergence term in the balance of linear momentum in the physical body $\delta\left(\operatorname{Div}\mathbf{P}\right)+\rho_0\delta\mathbf{B}=\rho_0\mathbf{A}$, has the following component form
\begin{equation}
	\delta\left(\operatorname{Div}\mathbf{P}\right)=\operatorname{Div}\delta\mathbf{P}
	=\operatorname{Div}\left(\boldsymbol{\mathsf{A}}:\nabla\mathbf{U}
	+\boldsymbol{\mathsf{B}}:\nabla\nabla\mathbf{U}\right)
	=\left(\mathsf{A}^{aA}{}_b{}^B~U^b{}_{|B}+\mathsf{B}^{aA}{}_b{}^{BC}~U^b{}_{|B|C}\right)_{|A}
	\frac{\partial}{\partial x^a}.
\end{equation}
Under a cloaking transformation $\Xi:\mathcal{B}\rightarrow\tilde{\mathcal{B}}$ it is transformed to \citep{Yavari2019}
\begin{equation}
	J_{\Xi}\left(\tilde{\mathsf{A}}^{\tilde{a}\tilde{A}}{}_{\tilde{b}}{}^{\tilde{B}}~\tilde{U}^{\tilde{b}}{}_{|\tilde{B}}
	+\tilde{\mathsf{B}}^{\tilde{a}\tilde{A}}{}_{\tilde{b}}{}^{\tilde{B}\tilde{C}}~
	\tilde{U}^{\tilde{b}}{}_{|\tilde{B}|\tilde{C}}\right)_{|\tilde{A}}\frac{\partial}{\partial \tilde{x}^{\tilde{a}}},
\end{equation}
where
\begin{equation} \label{Transformed-A-B1}
\begin{aligned}
	\tilde{U}^{\tilde{a}} &=\mathsf{s}^{\tilde{a}}{}_a U^a\,, \\
	\tilde{\mathsf{A}}^{\tilde{a}\tilde{A}}{}_{\tilde{b}}{}^{\tilde{B}} &=J_{\Xi}^{-1}\mathsf{s}^{\tilde{a}}{}_a
	\cloak{F}^{\tilde{A}}{}_A(\mathsf{s}^{-1})^b{}_{\tilde{b}}\cloak{F}^{\tilde{B}}{}_B~\mathsf{A}^{aA}{}_b{}^B
	+J_{\Xi}^{-1}\mathsf{s}^{\tilde{a}}{}_a	\cloak{F}^{\tilde{A}}{}_A(\mathsf{s}^{-1})^b{}_{\tilde{b}}
	\cloak{F}^{\tilde{B}}{}_{B|C}~\mathsf{B}^{aA}{}_b{}^{BC}\,, \\
	\tilde{\mathsf{B}}^{\tilde{a}\tilde{A}}{}_{\tilde{b}}{}^{\tilde{B}\tilde{C}} &=J_{\Xi}^{-1}\mathsf{s}^{\tilde{a}}{}_a
	\cloak{F}^{\tilde{A}}{}_A(\mathsf{s}^{-1})^b{}_{\tilde{b}}\cloak{F}^{\tilde{B}}{}_B\cloak{F}^{\tilde{C}}{}_C~
	\mathsf{B}^{aA}{}_b{}^{BC}\,.
\end{aligned}
\end{equation}
Note that the material and spatial coordinate charts for the virtual body are denoted by $\{\tilde{X}^{\tilde{A}}\}$, and $\{\tilde{x}^{\tilde{a}}\}$, respectively. 
Equivalently, \eqref{Transformed-A-B1} can be written as
\begin{equation}\label{Transformed-A-B}
\begin{aligned}
	\mathsf{A}^{aA}{}_b{}^B &=J_{\Xi}(\mathsf{s}^{-1})^a{}_{\tilde{a}}
	(\cloak{F}^{-1})^A{}_{\tilde{A}}\mathsf{s}^{\tilde{b}}{}_b(\cloak{F}^{-1})^B{}_{\tilde{B}}
	~\tilde{\mathsf{A}}^{\tilde{a}\tilde{A}}{}_{\tilde{b}}{}^{\tilde{B}}
	+J_{\Xi}(\mathsf{s}^{-1})^a{}_{\tilde{a}}(\cloak{F}^{-1})^A{}_{\tilde{A}}\mathsf{s}^{\tilde{b}}{}_b
	(\cloak{F}^{-1})^B{}_{\tilde{B}|\tilde{C}}~\tilde{\mathsf{B}}^{\tilde{a}\tilde{A}}{}_{\tilde{b}}{}^{\tilde{B}\tilde{C}}\,,
	 \\
	\mathsf{B}^{aA}{}_b{}^{BC} &=J_{\Xi}(\mathsf{s}^{-1})^a{}_{\tilde{a}}
	(\cloak{F}^{-1})^A{}_{\tilde{A}}\mathsf{s}^{\tilde{b}}{}_b(\cloak{F}^{-1})^B{}_{\tilde{B}}
	(\cloak{F}^{-1})^C{}_{\tilde{C}}~
	\tilde{\mathsf{B}}^{\tilde{a}\tilde{A}}{}_{\tilde{b}}{}^{\tilde{B}\tilde{C}}\,.
\end{aligned}
\end{equation}
It is assumed that $\mathring{F}^a{}_A=\delta^a_A$, and $\mathring{\tilde{F}}^{\tilde{a}}{}_{\tilde{A}}=\delta^{\tilde{a}}_{\tilde{A}}$. This implies that $\mathring{F}^a{}_{A|B}=0$, and $\mathring{\tilde{F}}^{\tilde{a}}{}_{\tilde{A}|\tilde{B}}=0$. It is also assumed that there is no initial stress in either configuration, i.e., $\mathring{P}^a{}_A=0$, $\mathring{H}^{aAB}=0$, and $\mathring{\tilde{P}}^{\tilde{a}}{}_{\tilde{A}}=0$, $\mathring{\tilde{H}}^{\tilde{a}\tilde{A}\tilde{B}}=0$.
From \eqref{Angular-Momentum-Gradient-Linear} the balance of angular momentum in the physical and virtual bodies read
\begin{eqnarray}
	\label{Angular-Momentum-Physical} \mathbb{A}^{[aM}{}_m{}^A\mathring{F}^{b]}{}_M=0, &
	~~~\mathbb{B}^{[aM}{}_m{}^{AB}\mathring{F}^{b]}{}_M=0,\\
	 \tilde{\mathbb{A}}^{[\tilde{a}\tilde{M}}{}_{\tilde{m}}{}^{\tilde{A}}
	\label{Angular-Momentum-Virtual} \mathring{\tilde{F}}^{\tilde{b}]}{}_{\tilde{M}}=0, & 
	~~~\tilde{\mathbb{B}}^{[\tilde{a}\tilde{M}}{}_{\tilde{m}}{}^{\tilde{A}\tilde{B}}
	\mathring{\tilde{F}}^{\tilde{b}]}{}_{\tilde{M}}=0.
\end{eqnarray}
For the uniform virtual body, from \eqref{Elastic-Constants-Gradient} one obtains
\begin{equation}
	\tilde{\mathsf{A}}^{\tilde{a}\tilde{A}}{}_{\tilde{b}}{}^{\tilde{B}}
	=\tilde{\mathbb{A}}^{\tilde{a}\tilde{A}}{}_{\tilde{b}}{}^{\tilde{B}},~~~~~~
	\tilde{\mathsf{B}}^{\tilde{a}\tilde{A}}{}_{\tilde{b}}{}^{\tilde{B}\tilde{C}}
	=\tilde{\mathbb{B}}^{\tilde{a}\tilde{A}}{}_{\tilde{b}}{}^{\tilde{B}\tilde{C}}-
	\tilde{\mathbb{B}}_{\tilde{b}}{}^{\tilde{B}\tilde{a}\tilde{A}\tilde{C}}.
\end{equation}
Thus, from \eqref{Angular-Momentum-Virtual}$_1$ one obtains
\begin{equation}
	\tilde{\mathsf{A}}^{[\tilde{a}\tilde{M}}{}_{\tilde{m}}{}^{\tilde{A}}\mathring{\tilde{F}}^{\tilde{b}]}{}_{\tilde{M}}=0,
	~~~\text{or}~~\tilde{\mathsf{A}}^{[\tilde{a}\tilde{b}]}{}_{\tilde{m}}{}^{\tilde{A}}=0.
\end{equation}
As for the physical body, from \eqref{Dynamics-Static}$_1$ we have
\begin{equation} 
	\mathbb{A}^{aAbB} = \mathbb{B}^{bBaAM}{}_{|M} + \mathsf{A}^{aAbB} \,.
\end{equation}
From the above relation and the balance of angular momentum \eqref{Angular-Momentum-Physical}${}_1$ one obtains
\begin{equation} \label{Angular-Momentum-Physical-Dynamic}
	\mathbb{B}^{bB[aA]M}{}_{|M} + \mathsf{A}^{[aA]bB}
	= 0 \,.
\end{equation}
Taking the antisymmetric part of the other pair of indices, i.e., $\mathbb{A}^{[aA][bB]}$, and since from \eqref{Angular-Momentum-Physical}${}_2$ one has $\mathbb{B}^{[bB][aA]C}{}_{|C}=0$, we obtain the following relations:
\begin{equation} \label{Constraint-Compact}
	\mathsf{A}^{[aA][bB]}
	= 0 \,.
\end{equation}
Now we are able to use the transformation \eqref{Transformed-A-B}$_1$. In particular, we make use of its fully contravariant version, viz.
\begin{equation} \label{Transformed-A-B-Contra}
	\mathsf{A}^{aAbB} =J_{\Xi}(\mathsf{s}^{-1})^a{}_{\tilde{a}}
	(\cloak{F}^{-1})^A{}_{\tilde{A}}(\mathsf{s}^{-1})^b{}_{\tilde{b}}\left[(\cloak{F}^{-1})^B{}_{\tilde{B}}
	~\tilde{\mathbb{A}}^{\tilde{a}\tilde{A}\tilde{b}\tilde{B}}+(\cloak{F}^{-1})^B{}_{\tilde{B}|\tilde{C}}
	~\tilde{\mathsf{B}}^{\tilde{a}\tilde{A}\tilde{b}\tilde{B}\tilde{C}}
	\right] \,.
\end{equation}
Note that in order to obtain \eqref{Transformed-A-B-Contra} we used the fact that $g_{\tilde{a}\tilde{b}}\, \mathsf s^{\tilde{b}}{}_b \, g^{ba} = (\mathsf s^{-1})^a{}_{\tilde{a}}$, which in turn comes from the fact that the shifter preserves the ambient space metric. Hence we can write \eqref{Constraint-Compact} as
\begin{equation} \label{Constraint1}
	(\mathsf{s}^{-1})^{[a}{}_{\tilde{a}}
	(\cloak{F}^{-1})^{A]}{}_{\tilde{A}}(\mathsf{s}^{-1})^{[b}{}_{\tilde{b}}
	\left[(\cloak{F}^{-1})^{B]}{}_{\tilde{B}}
	~\tilde{\mathbb{A}}^{\tilde{a}\tilde{A}\tilde{b}\tilde{B}}+(\cloak{F}^{-1})^{B]}{}_{\tilde{B}|\tilde{C}}
	~\tilde{\mathsf{B}}^{\tilde{a}\tilde{A}\tilde{b}\tilde{B}\tilde{C}} \right]=0 \,,
\end{equation}
or in the expanded form
\begin{equation} \label{Constraint2}
\begin{aligned}
	\left[(\mathsf{s}^{-1})^{a}{}_{\tilde{a}}(\cloak{F}^{-1})^{A}{}_{\tilde{A}}
	-(\mathsf{s}^{-1})^{A}{}_{\tilde{a}}(\cloak{F}^{-1})^{a}{}_{\tilde{A}}\right] 
	  & \Big\{ (\mathsf{s}^{-1})^{b}{}_{\tilde{b}}
	\left[(\cloak{F}^{-1})^{B}{}_{\tilde{B}}
	~\tilde{\mathbb{A}}^{\tilde{a}\tilde{A}\tilde{b}\tilde{B}}+(\cloak{F}^{-1})^{B}{}_{\tilde{B}|\tilde{C}}
	~\tilde{\mathsf{B}}^{\tilde{a}\tilde{A}\tilde{b}\tilde{B}\tilde{C}} \right] \\
	& -(\mathsf{s}^{-1})^{B}{}_{\tilde{b}}
	\left[(\cloak{F}^{-1})^{b}{}_{\tilde{B}}
	~\tilde{\mathbb{A}}^{\tilde{a}\tilde{A}\tilde{b}\tilde{B}}+(\cloak{F}^{-1})^{b}{}_{\tilde{B}|\tilde{C}}
	~\tilde{\mathsf{B}}^{\tilde{a}\tilde{A}\tilde{b}\tilde{B}\tilde{C}} \right]
	\Big\}=0\,.
\end{aligned}
\end{equation}
%
Note that \eqref{Constraint1}, i.e., $\mathbb{A}^{[aA][bB]}=0$, consists of six independent equations by virtue of the major symmetry \eqref{Symmetries}$_1$ for $\mathbb{A}$.
Albeit the static constants $\mathbb{A}^{aAbB}$ in the physical body must satisfy this property, it does not come automatically from the transformation~\eqref{Transformed-A-B-Contra}. This is in contrast with classical linearized elasticity \citep{Yavari2019}.
In transformation cloaking for gradient elasticity, the preservation of the balance of linear momentum gives a transformation in terms of the dynamic elastic constants $\mathsf{A}^{aAbB}$, and hence, the major symmetries of static constants $\mathbb{A}^{aAbB}$ for the physical problem are not immediate.
Therefore, the constraints \eqref{Constraint1} consist of nine equations.
Note that enforcing the major symmetry on $\mathbb{A}^{aAbB}$ in the physical body separately would not provide any useful equation besides an identity involving the derivatives of the tensor $\mathbb{B}$.


\begin{remark}
From \eqref{Angular-Momentum-Physical-Dynamic} one has $\mathbb{B}^{bB[aA]C}{}_{|C} = -\mathsf{A}^{[aA]bB}$ and hence from \eqref{Transformed-A-B-Contra}
\begin{equation} \label{B-PDE-constr}
	\mathbb{B}^{bB[aA]C}{}_{|C}
	=-J_{\Xi}(\mathsf{s}^{-1})^{[a}{}_{\tilde{a}}
	(\cloak{F}^{-1})^{A]}{}_{\tilde{A}}(\mathsf{s}^{-1})^b{}_{\tilde{b}}\left[(\cloak{F}^{-1})^B{}_{\tilde{B}}
	~\tilde{\mathbb{A}}^{\tilde{a}\tilde{A}\tilde{b}\tilde{B}}+(\cloak{F}^{-1})^B{}_{\tilde{B}|\tilde{C}}
	~\tilde{\mathsf{B}}^{\tilde{a}\tilde{A}\tilde{b}\tilde{B}\tilde{C}} \right]\,.
\end{equation}
Moreover, taking the divergence of \eqref{Dynamics-Static}$_2$ (applied to the elastic constants in the physical body) with respect to the index $C$, and the antisymmetric part with respect to the pair $aA$, one obtains
\begin{equation} 
	\mathbb{B}^{[aA]bBC}{}_{|C}-\mathbb{B}^{bB[aA]C}{}_{|C}=\mathsf{B}^{[aA]bBC}{}_{|C}+\mathbb{C}^{[aA]MbBC}{}_{|M|C} \,.
\end{equation}
By virtue of the balance of angular momentum in the physical body \eqref{Angular-Momentum-Physical}$_2$, one can then write
\begin{equation}
	-\mathbb{B}^{bB[aA]C}{}_{|C}=\mathsf{B}^{[aA]bBC}{}_{|C}+\mathbb{C}^{[aA]MbBC}{}_{|M|C} \,,
\end{equation}
and from \eqref{Angular-Momentum-Physical-Dynamic}
\begin{equation} \label{C-B-dynamicA}
	\mathbb{C}^{[aA]MbBC}{}_{|M|C} =- \mathsf{B}^{[aA]bBC}{}_{|C} - \mathsf{A}^{[aA]bB} \,.
\end{equation}
Note that from \eqref{Transformed-A-B}$_2$
\begin{equation}
	\mathsf{B}^{aAbBC} =J_{\Xi}(\mathsf{s}^{-1})^a{}_{\tilde{a}}
	(\cloak{F}^{-1})^A{}_{\tilde{A}}(\mathsf{s}^{-1})^b{}_{\tilde{b}}(\cloak{F}^{-1})^B{}_{\tilde{B}}
	(\cloak{F}^{-1})^C{}_{\tilde{C}}~
	\tilde{\mathsf{B}}^{\tilde{a}\tilde{A}\tilde{b}\tilde{B}\tilde{C}} \,,
\end{equation}
and therefore \eqref{C-B-dynamicA} becomes
\begin{equation} \label{C-PDE-constr}
\begin{aligned}
	\mathbb{C}^{[aA]NbBM}{}_{|M|N} =
	& \left[ J_{\Xi}(\mathsf{s}^{-1})^{[a}{}_{\tilde{a}}
	(\cloak{F}^{-1})^{A]}{}_{\tilde{A}}(\mathsf{s}^{-1})^b{}_{\tilde{b}}(\cloak{F}^{-1})^B{}_{\tilde{B}}
	(\cloak{F}^{-1})^M{}_{\tilde{M}}~
	\tilde{\mathsf{B}}^{\tilde{a}\tilde{A}\tilde{b}\tilde{B}\tilde{M}}  \right]_{|M}\\
	& +J_{\Xi}(\mathsf{s}^{-1})^{[a}{}_{\tilde{a}}
	(\cloak{F}^{-1})^{A]}{}_{\tilde{A}}(\mathsf{s}^{-1})^b{}_{\tilde{b}}\left[(\cloak{F}^{-1})^B{}_{\tilde{B}}
	~\tilde{\mathbb{A}}^{\tilde{a}\tilde{A}\tilde{b}\tilde{B}}+(\cloak{F}^{-1})^B{}_{\tilde{B}|\tilde{C}}
	~\tilde{\mathsf{B}}^{\tilde{a}\tilde{A}\tilde{b}\tilde{B}\tilde{C}}
	\right].
\end{aligned}
\end{equation}
Eqs. \eqref{B-PDE-constr} and \eqref{C-PDE-constr} represent differential constraints for $\mathbb{B}^{aAbBC}$ and $\mathbb{C}^{aANbBM}{}_{|M|N}$, respectively, and are a consequence of the balance of angular momentum.
\end{remark}

Next, we assume that the virtual body is isotropic and non-centrosymmetric. Knowing that $\mathring{\tilde{F}}^{\tilde{a}}{}_{\tilde{M}}=\delta^{\tilde{a}}_{\tilde{M}}$, with an abuse of notation from \eqref{L-B-Relation} one can write
\begin{equation}
	\tilde{\mathbb{B}}^{\tilde{a}\tilde{A}\tilde{b}\tilde{B}\tilde{C}}=
	4\tilde{\mathbb{L}}^{\tilde{A}\tilde{a}\tilde{b}\tilde{B}\tilde{C}}=4\tilde{\mathbb{L}}^{\tilde{a}\tilde{A}\tilde{b}\tilde{B}\tilde{C}}.
\end{equation}
Therefore, for the isotropic virtual body in Cartesian coordinates
\begin{equation} \label{B-Cartesian}
	\tilde{\mathbb{B}}^{\tilde{a}\tilde{A}\tilde{b}\tilde{B}\tilde{C}}=
	4\tilde{b}_0\left(\epsilon^{\tilde{a}\tilde{b}\tilde{C}}\delta^{\tilde{A}\tilde{B}}
	+\epsilon^{\tilde{A}\tilde{b}\tilde{C}}\delta^{\tilde{a}\tilde{B}}
	+\epsilon^{\tilde{a}\tilde{B}\tilde{C}}\delta^{\tilde{A}\tilde{b}}
	+\epsilon^{\tilde{A}\tilde{B}\tilde{C}}\delta^{\tilde{a}\tilde{b}}\right).
\end{equation}
In arbitrary curvilinear coordinates, one has
\begin{equation} \label{B-curvilinear}
	 \tilde{\mathbb{B}}^{\tilde{a}\tilde{A}\tilde{b}\tilde{B}\tilde{C}} =
	4\tilde{b}_0\left(\varepsilon^{\tilde{a}\tilde{b}\tilde{C}}g^{\tilde{A}\tilde{B}}
	+\varepsilon^{\tilde{A}\tilde{b}\tilde{C}}g^{\tilde{a}\tilde{B}}
	+\varepsilon^{\tilde{a}\tilde{B}\tilde{C}}g^{\tilde{A}\tilde{b}}
	+\varepsilon^{\tilde{A}\tilde{B}\tilde{C}}g^{\tilde{a}\tilde{b}}\right),
\end{equation}
where $\varepsilon^{\tilde{a}\tilde{b}\tilde{c}}=\frac{1}{\sqrt{g}}\epsilon^{\tilde{a}\tilde{b}\tilde{c}}$, and $g=\operatorname{det}\mathbf{g}$.
Note that the balance of angular momentum, i.e., $\tilde{\mathbb{B}}^{[\tilde{a}\tilde{A}]\tilde{b}\tilde{B}\tilde{C}}=0$ is satisfied.
In Cartesian coordinates, since the representation \eqref{B-Cartesian} is such that $\tilde{\mathbb{B}}^{\tilde{a}\tilde{A}\tilde{b}\tilde{B}\tilde{C}} = -\tilde{\mathbb{B}}^{\tilde{b}\tilde{B}\tilde{a}\tilde{A}\tilde{C}}$, the dynamic elastic constants are written as $\tilde{\mathsf{B}}^{\tilde{a}\tilde{A}\tilde{b}\tilde{B}\tilde{C}}=\tilde{\mathbb{B}}^{\tilde{a}\tilde{A}\tilde{b}\tilde{B}\tilde{C}}-\tilde{\mathbb{B}}^{\tilde{b}\tilde{B}\tilde{a}\tilde{A}\tilde{C}} = 2 \tilde{\mathbb{B}}^{\tilde{a}\tilde{A}\tilde{b}\tilde{B}\tilde{C}}$ and hence
\begin{equation}
	 \tilde{\mathsf{B}}^{\tilde{a}\tilde{A}\tilde{b}\tilde{B}\tilde{C}} =
	8\tilde{b}_0\left(\epsilon^{\tilde{a}\tilde{b}\tilde{C}}\delta^{\tilde{A}\tilde{B}}
	+\epsilon^{\tilde{A}\tilde{b}\tilde{C}}\delta^{\tilde{a}\tilde{B}}
	+\epsilon^{\tilde{a}\tilde{B}\tilde{C}}\delta^{\tilde{A}\tilde{b}}
	+\epsilon^{\tilde{A}\tilde{B}\tilde{C}}\delta^{\tilde{a}\tilde{b}}\right).
\end{equation}
In curvilinear coordinates
\begin{equation} \label{B-isotropic}
	 \tilde{\mathsf{B}}^{\tilde{a}\tilde{A}\tilde{b}\tilde{B}\tilde{C}} =
	8\tilde{b}_0\left(\varepsilon^{\tilde{a}\tilde{b}\tilde{C}}g^{\tilde{A}\tilde{B}}
	+\varepsilon^{\tilde{A}\tilde{b}\tilde{C}}g^{\tilde{a}\tilde{B}}
	+\varepsilon^{\tilde{a}\tilde{B}\tilde{C}}g^{\tilde{A}\tilde{b}}
	+\varepsilon^{\tilde{A}\tilde{B}\tilde{C}}g^{\tilde{a}\tilde{b}}\right).
\end{equation}
From the compatibility of $\accentset{\Xi}{\mathbf{F}}$ we know that $(\cloak{F}^{-1})^B{}_{[\tilde{B}|\tilde{C}]}=0$.
Thus, in curvilinear coordinates
\begin{equation} 
	(\cloak{F}^{-1})^B{}_{\tilde{B}|\tilde{C}}~\tilde{\mathsf{B}}^{\tilde{a}\tilde{A}\tilde{b}\tilde{B}\tilde{C}}=
	8\tilde{b}_0(\cloak{F}^{-1})^B{}_{\tilde{B}|\tilde{C}} \left(\varepsilon^{\tilde{a}\tilde{b}\tilde{C}} g^{\tilde{A}\tilde{B}}
	+\varepsilon^{\tilde{A}\tilde{b}\tilde{C}} g^{\tilde{a}\tilde{B}}\right).
\end{equation}
Moreover, with the usual abuse of notation for the indices, in the isotropic case one has the representation
\begin{equation} \label{A-isotropic}
	\mathbb{A}^{\tilde{a}\tilde{A}\tilde{b}\tilde{B}} = \lambda g^{\tilde{a}\tilde{A}} g^{\tilde{b}\tilde{B}}
	+ \mu \left(g^{\tilde{a}\tilde{b}} g^{\tilde{A}\tilde{B}} + g^{\tilde{a}\tilde{B}} g^{\tilde{A}\tilde{b}}\right) \,,
\end{equation}
so that \eqref{Constraint1} becomes
\begin{equation} \label{Constraint-Isotropic}
\begin{aligned}
	(\mathsf{s}^{-1})^{[a}{}_{\tilde{a}}
	(\cloak{F}^{-1})^{A]}{}_{\tilde{A}}(\mathsf{s}^{-1})^{[b}{}_{\tilde{b}} 
	\Big[ &(\cloak{F}^{-1})^{B]}{}_{\tilde{B}}
	\left( \lambda g^{\tilde{a}\tilde{b}} g^{\tilde{b}\tilde{B}}
	+ \mu g^{\tilde{a}\tilde{b}} g^{\tilde{A}\tilde{B}}
	+ \mu g^{\tilde{a}\tilde{B}} g^{\tilde{A}\tilde{b}} \right) \\
	&+8\tilde{b}_0(\cloak{F}^{-1})^B{}_{\tilde{B}|\tilde{C}} \left(\varepsilon^{\tilde{a}\tilde{b}\tilde{C}} g^{\tilde{A}\tilde{B}}
	+\varepsilon^{\tilde{A}\tilde{b}\tilde{C}} g^{\tilde{a}\tilde{B}}\right)
	 \Big]=0 \,.
\end{aligned}
\end{equation}

\subsection{Circular cylindrical and spherical cloaks}

We work with cylindrical $(R,\Theta,Z)$ and spherical $(R,\Theta,\Phi)$ coordinates, with $\Theta$ and $\Phi$ being the azimuthal and polar angles, respectively.
In both cases, the cloaking map is radial, and is represented by a function $\tilde R=f(R)$, so that one has $\cloak{\mathbf{F}}=\operatorname{diag}(f'(R),1,1)$.
As for the shifters, from \eqref{s-cylindrical} one obtains $\boldsymbol{\mathsf{s}}=\operatorname{diag}(1,R/f(R),1)$ for the cylindrical case, while
\eqref{s-spherical} gives the spherical case $\boldsymbol{\mathsf{s}}=\operatorname{diag}(1,R/f(R),R/f(R))$.
The metric tensor in the physical body has the representations $\mathbf{g}=\operatorname{diag}(1,R^2,1)$ and $\mathbf{g}=\operatorname{diag}(1,R^2,R^2\sin^2\Theta)$ in cylindrical and spherical coordinates, respectively, while the ones in the virtual body read $\tilde{\mathbf{g}}=\operatorname{diag}(1,f(R)^2,1)$ and $\tilde{\mathbf{g}}=\operatorname{diag}(1,f(R)^2,f(R)^2\sin^2\Theta)$.
We show that the conditions \eqref{Constraint2} cannot be satisfied for either a circular cylindrical or a spherical cloak when a radial cloaking map is utilized (see Fig. \ref{fig:circular}). Let us expand \eqref{Constraint-Isotropic} for $a=b=1$, and $A=B=3$:
\begin{equation}
\begin{aligned}
	\left[(\mathsf{s}^{-1})^{1}{}_{\tilde{a}}(\cloak{F}^{-1})^{3}{}_{\tilde{A}}
	-(\mathsf{s}^{-1})^{3}{}_{\tilde{a}}(\cloak{F}^{-1})^{1}{}_{\tilde{A}}\right] 
	  & \Big\{ (\mathsf{s}^{-1})^{1}{}_{\tilde{b}}
	\left[(\cloak{F}^{-1})^{3}{}_{\tilde{B}}
	~\tilde{\mathbb{A}}^{\tilde{a}\tilde{A}\tilde{b}\tilde{B}}+(\cloak{F}^{-1})^{3}{}_{\tilde{B}|\tilde{C}}
	~\tilde{\mathsf{B}}^{\tilde{a}\tilde{A}\tilde{b}\tilde{B}\tilde{C}} \right] \\
	& -(\mathsf{s}^{-1})^{3}{}_{\tilde{b}}
	\left[(\cloak{F}^{-1})^{1}{}_{\tilde{B}}
	~\tilde{\mathbb{A}}^{\tilde{a}\tilde{A}\tilde{b}\tilde{B}}+(\cloak{F}^{-1})^{1}{}_{\tilde{B}|\tilde{C}}
	~\tilde{\mathsf{B}}^{\tilde{a}\tilde{A}\tilde{b}\tilde{B}\tilde{C}} \right]
	\Big\}=0\,.
\end{aligned}
\end{equation}
Knowing that in the spherical (or cylindrical)  coordinates and for a radial cloaking map $\boldsymbol{\mathsf{s}}^{-1}$ and $\cloak{\mathbf{F}}^{-1}$ have diagonal representations, the above relation is simplified to read 
\begin{equation}
\begin{aligned}
	\left[(\mathsf{s}^{-1})^{1}{}_{\tilde{a}}(\cloak{F}^{-1})^{3}{}_{\tilde{A}}
	-(\mathsf{s}^{-1})^{3}{}_{\tilde{a}}(\cloak{F}^{-1})^{1}{}_{\tilde{A}}\right] 
	  & \Big\{ (\mathsf{s}^{-1})^{1}{}_{1}
	\left[(\cloak{F}^{-1})^{3}{}_{3}
	~\tilde{\mathbb{A}}^{\tilde{a}\tilde{A}13}+(\cloak{F}^{-1})^{3}{}_{3|\tilde{C}}
	~\tilde{\mathsf{B}}^{\tilde{a}\tilde{A}13\tilde{C}} \right] \\
	& -(\mathsf{s}^{-1})^{3}{}_{3}
	\left[(\cloak{F}^{-1})^{1}{}_{1}
	~\tilde{\mathbb{A}}^{\tilde{a}\tilde{A}31}+(\cloak{F}^{-1})^{1}{}_{1|\tilde{C}}
	~\tilde{\mathsf{B}}^{\tilde{a}\tilde{A}31\tilde{C}} \right]
	\Big\}=0\,.
\end{aligned}
\end{equation}
Hence
\begin{equation}
\begin{aligned}
	(\mathsf{s}^{-1})^{1}{}_{1}(\cloak{F}^{-1})^{3}{}_{3}
	  & \Big\{ (\mathsf{s}^{-1})^{1}{}_{1}
	\left[(\cloak{F}^{-1})^{3}{}_{3}
	~\tilde{\mathbb{A}}^{1313}+(\cloak{F}^{-1})^{3}{}_{3|\tilde{C}}
	~\tilde{\mathsf{B}}^{1313\tilde{C}} \right] \\
	& -(\mathsf{s}^{-1})^{3}{}_{3}
	\left[(\cloak{F}^{-1})^{1}{}_{1}
	~\tilde{\mathbb{A}}^{1331}+(\cloak{F}^{-1})^{1}{}_{1|\tilde{C}}
	~\tilde{\mathsf{B}}^{1331\tilde{C}} \right]
	\Big\} \\
	 -(\mathsf{s}^{-1})^{3}{}_{3}(\cloak{F}^{-1})^{1}{}_{1}  
	 & \Big\{ (\mathsf{s}^{-1})^{1}{}_{1}
	\left[(\cloak{F}^{-1})^{3}{}_{3}
	~\tilde{\mathbb{A}}^{3113}+(\cloak{F}^{-1})^{3}{}_{3|\tilde{C}}
	~\tilde{\mathsf{B}}^{3113\tilde{C}} \right] \\
	& -(\mathsf{s}^{-1})^{3}{}_{3}
	\left[(\cloak{F}^{-1})^{1}{}_{1}
	~\tilde{\mathbb{A}}^{3131}+(\cloak{F}^{-1})^{1}{}_{1|\tilde{C}}
	~\tilde{\mathsf{B}}^{3131\tilde{C}} \right]
	\Big\}
	=0\,.
\end{aligned}
\end{equation}

\begin{figure}[hbt!]
\centering
\vskip 0.1in
\includegraphics[width=0.9\textwidth]{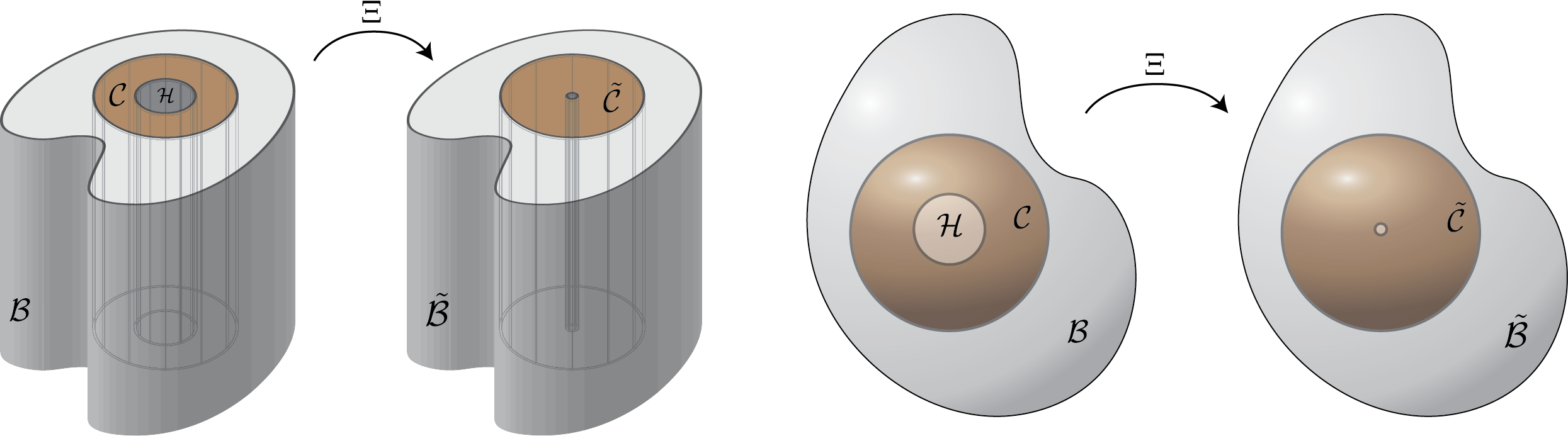}
\vskip 0.1in
\caption{Circular cylindrical (left) and spherical (right) holes and cloaks. The cloaking maps are assumed to be radially symmetric.}
\label{fig:circular}
\end{figure}
%
Note that from \eqref{A-isotropic}, and from the expression of the metric $\tilde{\mathbf{g}}$ in both cylindrical and spherical coordinates, one has $\tilde{\mathbb{A}}^{1313}=\tilde{\mathbb{A}}^{1331}=\tilde{\mathbb{A}}^{3113}=\tilde{\mathbb{A}}^{3131}=\tilde{\mu}/C>0$, with $C=1$ and $C=f(R)^2\sin^2\Theta$ in cylindrical and spherical cloaking, respectively. Moreover, noting that the metric tensor in both cylindrical and spherical coordinates is diagonal, from \eqref{B-curvilinear} one can easy see that $\tilde{\mathsf{B}}^{1313\tilde{C}}=\tilde{\mathsf{B}}^{1331\tilde{C}}=\tilde{\mathsf{B}}^{3113\tilde{C}}=\tilde{\mathsf{B}}^{3131\tilde{C}}=0$. Therefore\footnote{This is identical to the corresponding relation for centrosymmetric gradient solids investigated in \citep{Yavari2019}.}
\begin{equation}
\begin{aligned}
	(\mathsf{s}^{-1})^{1}{}_{1}(\cloak{F}^{-1})^{3}{}_{3}
	  & \left[ (\mathsf{s}^{-1})^{1}{}_{1} (\cloak{F}^{-1})^{3}{}_{3}  
	  -(\mathsf{s}^{-1})^{3}{}_{3} (\cloak{F}^{-1})^{1}{}_{1} \right] \\
	 -(\mathsf{s}^{-1})^{3}{}_{3}(\cloak{F}^{-1})^{1}{}_{1}  
	 & \left[ (\mathsf{s}^{-1})^{1}{}_{1}
	(\cloak{F}^{-1})^{3}{}_{3}-(\mathsf{s}^{-1})^{3}{}_{3} (\cloak{F}^{-1})^{1}{}_{1}  \right]
	=0\,.
\end{aligned}
\end{equation}
Therefore, $(\cloak{F}^{-1})^3{}_{3}(\mathsf{s}^{-1})^1{}_{1}=(\mathsf{s}^{-1})^3{}_{3}(\cloak{F}^{-1})^1{}_{1}$.
As $\cloak{\mathbf{F}}^{-1}=\operatorname{diag}(1/f'(R),1,1)$, and $\boldsymbol{\mathsf{s}}^{-1}=\operatorname{diag}\left(1,f(R)/R,1\right)$ and $\boldsymbol{\mathsf{s}}^{-1}=\operatorname{diag}\left(1,f(R)/R,f(R)/R\right)$, in the cylindrical and spherical coordinates, respectively, one must have $f(R)=R$, i.e., $\Xi=\mathrm{id}$. This means that cloaking is not possible.


\subsection{Spheroidal cloaks}
 
Next we consider prolate and oblate spheroidal holes and consider cloaking maps that respect the sphenoidal symmetry in the sense that they map a spheroid to another confocal spheroid (see Fig. \ref{fig:spheroids}). This will be a generalization of the spherical cloak problem.

\begin{figure}[hbt!]
\centering
\vskip 0.2in
\includegraphics[width=0.9\textwidth]{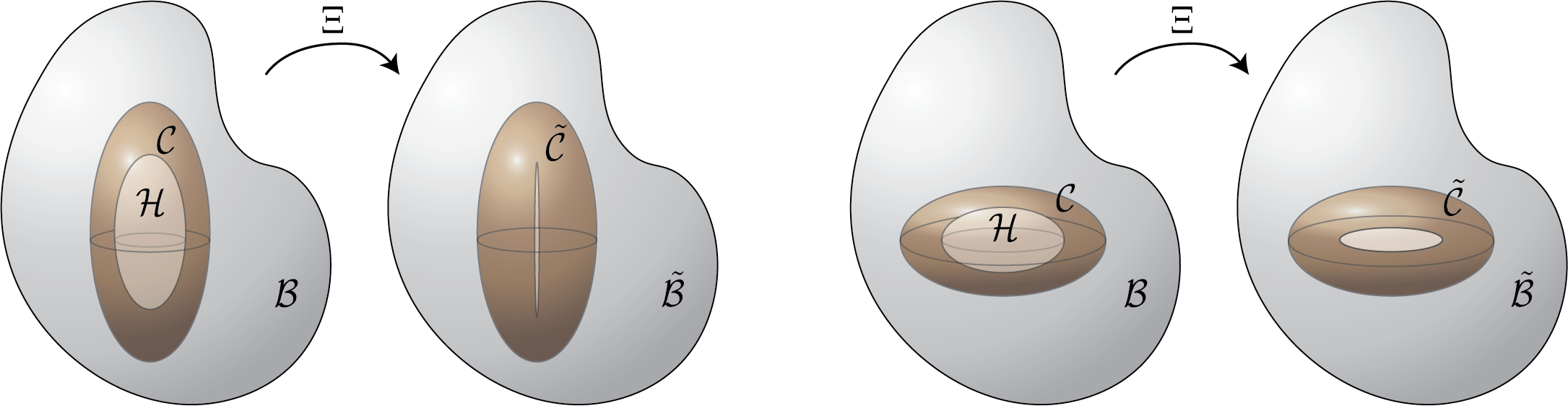}
\vskip 0.1in
\caption{Prolate (left) and oblate (right) spheroidal holes and cloaks. The cloaking maps are assumed to be spheroidally-symmetric. The cloaking map on the left shrinks a spheroidal hole to a needle-like hole. The one on the right shrinks the spheroidal hole to a disk-shaped hole.}
\label{fig:spheroids}
\end{figure}

\begin{prop}
Assuming that the virtual body is isotropic and non-centrosymmetric, elastodynamic transformation cloaking is not possible for either prolate or oblate spheroidal holes using any spheroidally-symmetric cloaking map.
\end{prop}
\begin{proof}
Let us consider a prolate spheroidal hole with focal distance $a$. The natural coordinates are the prolate spheroidal coordinates $(H,\Theta,\Phi)$, $0\leq H < \infty,~0\leq\Theta\leq\pi,~0\leq\Theta\leq2\pi$ defined as \citep{Moon2012}
\begin{equation}
\begin{cases}
	X=a\sinh H \sin \Theta \cos \Phi, \\
	Y=a\sinh H \sin \Theta \sin \Phi, \\
	Z=a\cosh H \cos \Theta.
\end{cases}
\end{equation} 
Note that $H = \mathrm{const}$ are prolate spheroids. We consider a cloaking map of the form $(\tilde{H},\tilde{\Theta},\tilde{\Phi})=\Xi(H,\Theta,\Phi)=(f(H),\Theta,\Phi)$. The shifter map reads
\begin{equation}
	\boldsymbol{\mathsf{s}}=
\begin{bmatrix} 
	2\frac{ \sinh H \sinh \tilde{H} \cos ^2\Theta
	+\cosh H \cosh \tilde{H} \sin ^2\Theta }{\cosh 2 \tilde{H}-\cos 2\Theta } 
   & \frac{2 \sin \Theta \cos \Theta \sinh (H-\tilde{H})}{\cosh 2 \tilde{H}-\cos 2\Theta} & 0 \\
 \frac{\sin 2\Theta \sinh (H-\tilde{H})}{\cos 2\Theta-\cosh 2\tilde{H} } &
   2\frac{\sinh H \sinh \tilde{H} \cos^2\Theta+\cosh H \cosh \tilde{H}  \sin^2\Theta }
   {\cosh 2 \tilde{H}-\cos 2\Theta } & 0 \\
 0 & 0 & \sinh H \operatorname{csch}\tilde{H}
\end{bmatrix} \,.
\end{equation}
The spatial metric has the following representation
\begin{equation}
\begin{aligned}
	& \mathbf{g}=
\begin{bmatrix} 
	a^2(\sinh^2H+\sin^2\Theta) & 0 & 0 \\
	0 & a^2(\sinh^2H+\sin^2\Theta) & 0  \\
        0 & 0 & a^2\sinh^2H\sin^2\Theta
\end{bmatrix}  \,.
\end{aligned}
\end{equation}
The cloaking derivative map has the coordinate representation
\begin{equation}
 \cloak{\mathbf{F}}=
\begin{bmatrix} 
	f'(H) & 0 & 0 \\
	0 & 1 & 0  \\
        0 & 0 & 1
\end{bmatrix}  \,.
\end{equation}
The $(a,A,b,B)=(2,3,2,3)$ component of the constraint \eqref{Constraint-Isotropic} reads
\begin{equation}
	\frac{\mu\left[\coth H-\coth(f(H))\right]^2}
	{a^4 (\cos 2\Theta-\cosh 2H) \left[\cos 2\Theta-\cosh(2 f(H))\right]}=0 \,.
\end{equation}
Therefore, $\coth(f(H))=\coth H$, or $f(H)=H$, i.e., cloaking is not possible.
 
In the case of an oblate spheroidal hole, one uses the oblate spheroidal coordinates $(H,\Theta,\Phi)$, $0\leq H < \infty,~0\leq\Theta\leq\pi,~0\leq\Theta\leq2\pi$ defined as \citep{Moon2012}
\begin{equation}
\begin{cases}
	X=a\cosh H \sin \Theta \cos \Phi, \\
	Y=a\cosh H \sin \Theta \sin \Phi, \\
	Z=a\sinh H \cos \Theta.
\end{cases}
\end{equation} 
The spatial metric has the following representation
\begin{equation}
\begin{aligned}
	& \mathbf{g}=
\begin{bmatrix} 
	a^2(\cosh^2H-\sin^2\Theta) & 0 & 0 \\
	0 & a^2(\cosh^2H+\sin^2\Theta) & 0  \\
        0 & 0 & a^2\cosh^2H\sin^2\Theta
\end{bmatrix}  \,.
\end{aligned}
\end{equation}
We again consider a cloaking map of the form $(\tilde{H},\tilde{\Theta},\tilde{\Phi})=\Xi(H,\Theta,\Phi)=(f(H),\Theta,\Phi)$.
In this case, the $(a,A,b,B)=(2,3,2,3)$ component of the constraint \eqref{Constraint-Isotropic} reads
\begin{equation}
	\frac{\mu\left[\tanh H-\tanh(f(H))\right]^2}
	{a^4 (\cos 2\Theta+\cosh 2H) \left[\cos 2\Theta+\cosh(2 f(H))\right]}=0 \,.
\end{equation}
Therefore, $\tanh(f(H))=\tanh H$, or $f(H)=H$, i.e., cloaking is not possible.
\end{proof}



\subsection{Non-symmetric cloaks}

Now one may ask whether cloaking would be possible for less symmetric holes and cloaking maps. 
From \eqref{Constraint2} we have
%
%
%
\begin{equation} 
\begin{aligned}
	\left[\delta^{a}_{\tilde{a}}(\cloak{F}^{-1})^{A}{}_{\tilde{A}}
	-\delta^{A}_{\tilde{a}}(\cloak{F}^{-1})^{a}{}_{\tilde{A}}\right] 
	  & \Big\{\delta^{b}_{\tilde{b}}\left[(\cloak{F}^{-1})^{B}{}_{\tilde{B}}
	~\tilde{\mathbb{A}}^{\tilde{a}\tilde{A}\tilde{b}\tilde{B}}+(\cloak{F}^{-1})^{B}{}_{\tilde{B}|\tilde{C}}
	~\tilde{\mathsf{B}}^{\tilde{a}\tilde{A}\tilde{b}\tilde{B}\tilde{C}} \right] \\
	& -\delta^{B}_{\tilde{b}}
	\left[(\cloak{F}^{-1})^{b}{}_{\tilde{B}}
	~\tilde{\mathbb{A}}^{\tilde{a}\tilde{A}\tilde{b}\tilde{B}}+(\cloak{F}^{-1})^{b}{}_{\tilde{B}|\tilde{C}}
	~\tilde{\mathsf{B}}^{\tilde{a}\tilde{A}\tilde{b}\tilde{B}\tilde{C}} \right]
	\Big\}=0\,.
\end{aligned}
\end{equation}
Or
\begin{equation} 
\begin{aligned}
	\left[\delta^{a}_{\tilde{a}}(\cloak{F}^{-1})^{A}{}_{\tilde{A}}
	-\delta^{A}_{\tilde{a}}(\cloak{F}^{-1})^{a}{}_{\tilde{A}}\right] 
	  & \Big\{\left[(\cloak{F}^{-1})^{B}{}_{\tilde{B}}
	~\tilde{\mathbb{A}}^{\tilde{a}\tilde{A}b\tilde{B}}+(\cloak{F}^{-1})^{B}{}_{\tilde{B}|\tilde{C}}
	~\tilde{\mathsf{B}}^{\tilde{a}\tilde{A}b\tilde{B}\tilde{C}} \right] \\
	& -\left[(\cloak{F}^{-1})^{b}{}_{\tilde{B}}
	~\tilde{\mathbb{A}}^{\tilde{a}\tilde{A}B\tilde{B}}+(\cloak{F}^{-1})^{b}{}_{\tilde{B}|\tilde{C}}
	~\tilde{\mathsf{B}}^{\tilde{a}\tilde{A}B\tilde{B}\tilde{C}} \right]
	\Big\}=0\,.
\end{aligned}
\end{equation}
Thus
\begin{equation} 
\begin{aligned}
	& \delta^{a}_{\tilde{a}}(\cloak{F}^{-1})^{A}{}_{\tilde{A}}
	\left\{\left[(\cloak{F}^{-1})^{B}{}_{\tilde{B}}
	~\tilde{\mathbb{A}}^{\tilde{a}\tilde{A}b\tilde{B}}+(\cloak{F}^{-1})^{B}{}_{\tilde{B}|\tilde{C}}
	~\tilde{\mathsf{B}}^{\tilde{a}\tilde{A}b\tilde{B}\tilde{C}} \right] 
	-\left[(\cloak{F}^{-1})^{b}{}_{\tilde{B}}
	~\tilde{\mathbb{A}}^{\tilde{a}\tilde{A}B\tilde{B}}+(\cloak{F}^{-1})^{b}{}_{\tilde{B}|\tilde{C}}
	~\tilde{\mathsf{B}}^{\tilde{a}\tilde{A}B\tilde{B}\tilde{C}} \right]
	\right\}\\
	&
	-\delta^{A}_{\tilde{a}}(\cloak{F}^{-1})^{a}{}_{\tilde{A}}
	  \Big\{\left[(\cloak{F}^{-1})^{B}{}_{\tilde{B}}
	~\tilde{\mathbb{A}}^{\tilde{a}\tilde{A}b\tilde{B}}+(\cloak{F}^{-1})^{B}{}_{\tilde{B}|\tilde{C}}
	~\tilde{\mathsf{B}}^{\tilde{a}\tilde{A}b\tilde{B}\tilde{C}} \right] 
	 -\left[(\cloak{F}^{-1})^{b}{}_{\tilde{B}}
	~\tilde{\mathbb{A}}^{\tilde{a}\tilde{A}B\tilde{B}}+(\cloak{F}^{-1})^{b}{}_{\tilde{B}|\tilde{C}}
	~\tilde{\mathsf{B}}^{\tilde{a}\tilde{A}B\tilde{B}\tilde{C}} \right]
	\Big\}
	=0\,.
\end{aligned}
\end{equation}
Hence
\begin{equation} \label{Constraint}
\begin{aligned}
	& (\cloak{F}^{-1})^{A}{}_{\tilde{A}}
	\left\{\left[(\cloak{F}^{-1})^{B}{}_{\tilde{B}}
	~\tilde{\mathbb{A}}^{a\tilde{A}b\tilde{B}}+(\cloak{F}^{-1})^{B}{}_{\tilde{B}|\tilde{C}}
	~\tilde{\mathsf{B}}^{a\tilde{A}b\tilde{B}\tilde{C}} \right] 
	-\left[(\cloak{F}^{-1})^{b}{}_{\tilde{B}}
	~\tilde{\mathbb{A}}^{a\tilde{A}B\tilde{B}}+(\cloak{F}^{-1})^{b}{}_{\tilde{B}|\tilde{C}}
	~\tilde{\mathsf{B}}^{a\tilde{A}B\tilde{B}\tilde{C}} \right]
	\right\}\\
	&
	-(\cloak{F}^{-1})^{a}{}_{\tilde{A}}
	  \Big\{\left[(\cloak{F}^{-1})^{B}{}_{\tilde{B}}
	~\tilde{\mathbb{A}}^{A\tilde{A}b\tilde{B}}+(\cloak{F}^{-1})^{B}{}_{\tilde{B}|\tilde{C}}
	~\tilde{\mathsf{B}}^{A\tilde{A}b\tilde{B}\tilde{C}} \right] 
	 -\left[(\cloak{F}^{-1})^{b}{}_{\tilde{B}}
	~\tilde{\mathbb{A}}^{A\tilde{A}B\tilde{B}}+(\cloak{F}^{-1})^{b}{}_{\tilde{B}|\tilde{C}}
	~\tilde{\mathsf{B}}^{A\tilde{A}B\tilde{B}\tilde{C}} \right]
	\Big\}
	=0\,.
\end{aligned}
\end{equation}
Note that if either $a=A$ or $b=B$ the above relations are trivial. We assume that $a\neq A$, and $b\neq B$. Let us consider an arbitrary cloaking transformation whose inverse derivative map $\cloak{\mathbf{F}}$ has the following representation in Cartesian coordinates
\begin{equation}
\cloak{\mathbf{F}}^{-1}=
\begin{bmatrix}
\mathsf{a}_{11} & \mathsf{a}_{12}  & \mathsf{a}_{13}  \\
\mathsf{a}_{21} & \mathsf{a}_{22}  & \mathsf{a}_{23}  \\
\mathsf{a}_{31} & \mathsf{a}_{32}  & \mathsf{a}_{33}  
\end{bmatrix}.
\end{equation} 
The covariant derivative of $\cloak{\mathbf{F}}^{-1}$ has the following representation:
\begin{equation}
\nabla\cloak{\mathbf{F}}^{-1}=
\renewcommand*{\arraystretch}{1.2}
\begin{bmatrix}
\begin{bmatrix}
\mathsf{F}_{111} & \mathsf{F}_{121}  & \mathsf{F}_{131}  \\
\mathsf{F}_{211} & \mathsf{F}_{221}  & \mathsf{F}_{231}  \\
\mathsf{F}_{311} & \mathsf{F}_{321}  & \mathsf{F}_{331}  
\end{bmatrix} \\
\begin{bmatrix}
\mathsf{F}_{112} & \mathsf{F}_{122}  & \mathsf{F}_{132}  \\
\mathsf{F}_{212} & \mathsf{F}_{222}  & \mathsf{F}_{232}  \\
\mathsf{F}_{312} & \mathsf{F}_{322}  & \mathsf{F}_{332}  
\end{bmatrix} \\
\begin{bmatrix}
\mathsf{F}_{113} & \mathsf{F}_{123}  & \mathsf{F}_{133}  \\
\mathsf{F}_{213} & \mathsf{F}_{223}  & \mathsf{F}_{233}  \\
\mathsf{F}_{313} & \mathsf{F}_{323}  & \mathsf{F}_{333}  
\end{bmatrix}
\end{bmatrix} .
\end{equation} 
Taking into account the compatibility equations, $\nabla\cloak{\mathbf{F}}^{-1}$ has the following representation:
\begin{equation}
\nabla\cloak{\mathbf{F}}^{-1}=
\renewcommand*{\arraystretch}{1.2}
\begin{bmatrix}
\begin{bmatrix}
\mathsf{F}_{111} & \mathsf{F}_{112}  & \mathsf{F}_{113}  \\
\mathsf{F}_{211} & \mathsf{F}_{212}  & \mathsf{F}_{213}  \\
\mathsf{F}_{311} & \mathsf{F}_{312}  & \mathsf{F}_{313}  
\end{bmatrix} \\
\begin{bmatrix}
\mathsf{F}_{112} & \mathsf{F}_{122}  & \mathsf{F}_{123}  \\
\mathsf{F}_{212} & \mathsf{F}_{222}  & \mathsf{F}_{223}  \\
\mathsf{F}_{312} & \mathsf{F}_{322}  & \mathsf{F}_{323}  
\end{bmatrix} \\
\begin{bmatrix}
\mathsf{F}_{113} & \mathsf{F}_{123}  & \mathsf{F}_{133}  \\
\mathsf{F}_{213} & \mathsf{F}_{223}  & \mathsf{F}_{233}  \\
\mathsf{F}_{313} & \mathsf{F}_{323}  & \mathsf{F}_{333}  
\end{bmatrix}
\end{bmatrix} .
\end{equation} 
As was mentioned earlier, \eqref{Constraint1}, or its equivalent variants \eqref{Constraint2}, \eqref{Constraint-Isotropic} and \eqref{Constraint}, provide a total of $81$ equations, of which only $9$ are independent.
When $\tilde{b}_0=0$ the number of equations reduces to $6$.
After plugging \eqref{A-isotropic} and \eqref{B-isotropic} into \eqref{Constraint}, one obtains the following nine independent equations:\footnote{Symbolic computations were done with Mathematica Version 12.0.0.0, Wolfram Research, Champaign, IL.}
\begin{equation} \label{Equation1}
\begin{aligned}
	& (3 \lambda+2 \mu)(\mathsf{a}_{12}-\mathsf{a}_{21})^2
	+ \mu \left[3(\mathsf{a}_{11}-\mathsf{a}_{22})^2+3(\mathsf{a}_{13}^2+\mathsf{a}_{23}^2)
	+3(\mathsf{a}_{12}+\mathsf{a}_{21})^2+(\mathsf{a}_{12}-\mathsf{a}_{21})^2\right] \\
    & +\beta  \Large[(\mathsf{a}_{11}-\mathsf{a}_{22}) (\mathsf{F}_{123}+\mathsf{F}_{213})
   +2 (\mathsf{a}_{12} \mathsf{F}_{223}- \mathsf{a}_{21} \mathsf{F}_{113}) \\
   & ~~~~ +\mathsf{a}_{23}(\mathsf{F}_{111}-\mathsf{F}_{133}+\mathsf{F}_{212})
   -\mathsf{a}_{13}(\mathsf{F}_{112}+\mathsf{F}_{222}-\mathsf{F}_{233})
   \Large]=0,
\end{aligned}
\end{equation} 
\begin{equation} \label{Equation2}
\begin{aligned}
	& (3 \lambda+2 \mu)(\mathsf{a}_{23}-\mathsf{a}_{32})^2
	+ \mu \left[3(\mathsf{a}_{22}-\mathsf{a}_{33})^2+3(\mathsf{a}_{21}^2+\mathsf{a}_{31}^2)
	+3(\mathsf{a}_{23}+\mathsf{a}_{32})^2+(\mathsf{a}_{23}-\mathsf{a}_{32})^2\right] \\
	&  +\beta \Large[
   (\mathsf{a}_{22}-\mathsf{a}_{33})(\mathsf{F}_{213}+\mathsf{F}_{312})
   +2(\mathsf{a}_{23} \mathsf{F}_{313}-\mathsf{a}_{32} \mathsf{F}_{212}) \\
   & ~~~~+\mathsf{a}_{31}(-\mathsf{F}_{211}+\mathsf{F}_{222}+\mathsf{F}_{323})
   -\mathsf{a}_{21} (\mathsf{F}_{223}-\mathsf{F}_{311}+\mathsf{F}_{333})
   \Large]=0,
\end{aligned}
\end{equation} 
\begin{equation} \label{Equation3}
\begin{aligned}
	& (3 \lambda+2 \mu)(\mathsf{a}_{13}-\mathsf{a}_{31})^2
	+ \mu \left[3(\mathsf{a}_{11}-\mathsf{a}_{33})^2+3(\mathsf{a}_{12}^2+\mathsf{a}_{32}^2)
	+3(\mathsf{a}_{13}+\mathsf{a}_{31})^2+(\mathsf{a}_{13}-\mathsf{a}_{31})^2\right] \\
   & +\beta \Large[(\mathsf{a}_{33}-\mathsf{a}_{11})(\mathsf{F}_{123}+\mathsf{F}_{312})
   +2(\mathsf{a}_{31} \mathsf{F}_{112}-\mathsf{a}_{13} \mathsf{F}_{323}) \\
   &~~~~ +\mathsf{a}_{12}(\mathsf{F}_{113}-\mathsf{F}_{322}+\mathsf{F}_{333})
   -\mathsf{a}_{32}(\mathsf{F}_{111}-\mathsf{F}_{122}+\mathsf{F}_{313})\Large] =0,
\end{aligned}
\end{equation} 
\begin{equation} \label{Equation4}
\begin{aligned}
	& 3 \mu \left[-\mathsf{a}_{11} (\mathsf{a}_{23}+\mathsf{a}_{32})+\mathsf{a}_{12} \mathsf{a}_{13}+2 \mathsf{a}_{21}
   \mathsf{a}_{31}+\mathsf{a}_{22} \mathsf{a}_{32}+\mathsf{a}_{23} \mathsf{a}_{33}\right]
   +3 \lambda (\mathsf{a}_{12}-\mathsf{a}_{21}) (\mathsf{a}_{13}-\mathsf{a}_{31}) \\
   & +\beta \Large[(\mathsf{a}_{11}-\mathsf{a}_{22})(\mathsf{F}_{111}-\mathsf{F}_{122}+\mathsf{F}_{313})
   +\mathsf{a}_{13} (\mathsf{F}_{113}-\mathsf{F}_{322}+\mathsf{F}_{333}) \\
   & ~~~~+2 \mathsf{a}_{12}(\mathsf{F}_{112}+\mathsf{F}_{323})
   +2\mathsf{a}_{21} \mathsf{F}_{112}+\mathsf{a}_{23} (\mathsf{F}_{123}+\mathsf{F}_{312}) \Large] =0,
\end{aligned}
\end{equation} 
\begin{equation} \label{Equation5}
\begin{aligned}
	&  3 \mu \left[-\mathsf{a}_{11} (\mathsf{a}_{23}+\mathsf{a}_{32})+\mathsf{a}_{12} \mathsf{a}_{13}+2 \mathsf{a}_{21}
   \mathsf{a}_{31}+\mathsf{a}_{22} \mathsf{a}_{32}+\mathsf{a}_{23} \mathsf{a}_{33}\right]
   +3 \lambda (\mathsf{a}_{12}-\mathsf{a}_{21}) (\mathsf{a}_{13}-\mathsf{a}_{31}) \\
   &+\beta \Large[(\mathsf{a}_{33}-\mathsf{a}_{11}) (\mathsf{F}_{111}-\mathsf{F}_{133}+\mathsf{F}_{212})-\mathsf{a}_{12}
   (\mathsf{F}_{112}+\mathsf{F}_{222}-\mathsf{F}_{233}) \\
   &~~~~-2 \mathsf{a}_{13} (\mathsf{F}_{113}+\mathsf{F}_{223})-2
   \mathsf{a}_{31} \mathsf{F}_{113}-\mathsf{a}_{32} (\mathsf{F}_{123}+\mathsf{F}_{213}) \Large]  =0,
\end{aligned}
\end{equation} 
\begin{equation} \label{Equation6}
\begin{aligned}
	&  -3 \mu \left[\mathsf{a}_{11} \mathsf{a}_{31}+2 \mathsf{a}_{12} \mathsf{a}_{32}-\mathsf{a}_{22}
   (\mathsf{a}_{13}+\mathsf{a}_{31})+\mathsf{a}_{13} \mathsf{a}_{33}+\mathsf{a}_{21} \mathsf{a}_{23} \right]
   +3 \lambda (\mathsf{a}_{12}-\mathsf{a}_{21}) (\mathsf{a}_{23}-\mathsf{a}_{32}) \\
   & +\beta \Large[(\mathsf{a}_{11}-\mathsf{a}_{22}) (\mathsf{F}_{211}-\mathsf{F}_{222}-\mathsf{F}_{323})
   +\mathsf{a}_{23} (\mathsf{F}_{223}-\mathsf{F}_{311}+\mathsf{F}_{333})\\
   &~~~~+2 \mathsf{a}_{12}\mathsf{F}_{212}+\mathsf{a}_{13} (\mathsf{F}_{213}+\mathsf{F}_{312})+2 \mathsf{a}_{21}
   (\mathsf{F}_{212}+\mathsf{F}_{313})\Large]=0,
\end{aligned}
\end{equation} 
\begin{equation} \label{Equation7}
\begin{aligned}
	&  -3 \mu \left[\mathsf{a}_{11} \mathsf{a}_{31}+2 \mathsf{a}_{12} \mathsf{a}_{32}-\mathsf{a}_{22}
   (\mathsf{a}_{13}+\mathsf{a}_{31})+\mathsf{a}_{13} \mathsf{a}_{33}+\mathsf{a}_{21} \mathsf{a}_{23}\right]
   +3 \lambda (\mathsf{a}_{12}-\mathsf{a}_{21}) (\mathsf{a}_{23}-\mathsf{a}_{32}) \\
   & +\beta \Large[-\mathsf{a}_{21}(\mathsf{F}_{111}-\mathsf{F}_{133}+\mathsf{F}_{212})-(\mathsf{a}_{22}-\mathsf{a}_{33})
   (\mathsf{F}_{112}+\mathsf{F}_{222}-\mathsf{F}_{233}) \\
   &~~~~-2 \mathsf{a}_{23}(\mathsf{F}_{113}+\mathsf{F}_{223})-\mathsf{a}_{31} (\mathsf{F}_{123}+\mathsf{F}_{213})
   -2 \mathsf{a}_{32}
   \mathsf{F}_{223}\Large]=0,
\end{aligned}
\end{equation} 
\begin{equation} \label{Equation8}
\begin{aligned}
	& 3 \mu \left[\mathsf{a}_{11} \mathsf{a}_{21}-\mathsf{a}_{33} (\mathsf{a}_{12}+\mathsf{a}_{21})
	+\mathsf{a}_{12} \mathsf{a}_{22}+2\mathsf{a}_{13} \mathsf{a}_{23}+\mathsf{a}_{31} \mathsf{a}_{32} \right]
   +3 \lambda  (\mathsf{a}_{13}-\mathsf{a}_{31}) (\mathsf{a}_{23}-\mathsf{a}_{32}) \\
   & +\beta \Large[(\mathsf{a}_{33}-\mathsf{a}_{11})(\mathsf{F}_{223}-\mathsf{F}_{311}+\mathsf{F}_{333})
   +\mathsf{a}_{32}(-\mathsf{F}_{211}+\mathsf{F}_{222}+\mathsf{F}_{323}) \\
   &~~~~+\mathsf{a}_{12} (\mathsf{F}_{213}+\mathsf{F}_{312})
   +2\mathsf{a}_{13} \mathsf{F}_{313}+2 \mathsf{a}_{31} (\mathsf{F}_{212}+\mathsf{F}_{313}) \Large] =0,
\end{aligned}
\end{equation} 
\begin{equation} \label{Equation9}
\begin{aligned}
	& 3 \mu \left[\mathsf{a}_{11} \mathsf{a}_{21}-\mathsf{a}_{33} (\mathsf{a}_{12}+\mathsf{a}_{21})
	+\mathsf{a}_{12} \mathsf{a}_{22}+2\mathsf{a}_{13} \mathsf{a}_{23}+\mathsf{a}_{31} \mathsf{a}_{32}\right]
   +3 \lambda  (\mathsf{a}_{13}-\mathsf{a}_{31}) (\mathsf{a}_{23}-\mathsf{a}_{32}) \\
   & +\beta \Large[(\mathsf{a}_{22}-\mathsf{a}_{33})(\mathsf{F}_{113}-\mathsf{F}_{322}+\mathsf{F}_{333})
   -\mathsf{a}_{31}(\mathsf{F}_{111}-\mathsf{F}_{122}+\mathsf{F}_{313}) \\
   &~~~~ -\mathsf{a}_{21}(\mathsf{F}_{123}+\mathsf{F}_{312}) 
   -2 \mathsf{a}_{23} \mathsf{F}_{323} -2 \mathsf{a}_{32} (\mathsf{F}_{112}+\mathsf{F}_{323}) \Large] =0,
\end{aligned}
\end{equation} 
where $\beta=8\tilde{b}_0$.
Note that $\tilde{b}_0^2$ is bounded by a product of $\tilde{\mu}$ and the sixth-order elastic constants \citep{Papanicolopulos2011}. In the above system of PDEs one can only assume that $\beta\neq 0$ as the sixth-order elastic constants do not appear in the constraints \eqref{Constraint}.

Subtracting Eq.\eqref{Equation5} from Eq.\eqref{Equation4}, Eq.\eqref{Equation7} from Eq.\eqref{Equation6}, and Eq.\eqref{Equation9} from Eq.\eqref{Equation8}, and assuming that $\beta\neq 0$ one obtains the following system of PDEs:
\begin{equation}
\begin{aligned}
   & -\mathsf{a}_{11} (\mathsf{F}_{223}-\mathsf{F}_{311}+\mathsf{F}_{333})
   -\mathsf{a}_{22} (\mathsf{F}_{113}-\mathsf{F}_{322}+\mathsf{F}_{333})
   +\mathsf{a}_{33}(\mathsf{F}_{113}+\mathsf{F}_{223}-\mathsf{F}_{311}-\mathsf{F}_{322}+2 \mathsf{F}_{333}) \\
   &  +\mathsf{a}_{31}(\mathsf{F}_{111}-\mathsf{F}_{122}+2 \mathsf{F}_{212}+3 \mathsf{F}_{313})
   +\mathsf{a}_{32} (2\mathsf{F}_{112}-\mathsf{F}_{211}+\mathsf{F}_{222}+3 \mathsf{F}_{323}) \\
   & +\mathsf{a}_{12} (\mathsf{F}_{213}+\mathsf{F}_{312})+\mathsf{a}_{21}(\mathsf{F}_{123}+\mathsf{F}_{312})
   +2(\mathsf{a}_{13} \mathsf{F}_{313}+\mathsf{a}_{23} \mathsf{F}_{323})=0,
\end{aligned}
\end{equation} 
\begin{equation}
\begin{aligned}
  & \mathsf{a}_{11} (\mathsf{F}_{211}-\mathsf{F}_{222}-\mathsf{F}_{323})
  -\mathsf{a}_{33} (\mathsf{F}_{112}+\mathsf{F}_{222}-\mathsf{F}_{233})  
  +\mathsf{a}_{22} (\mathsf{F}_{112}-\mathsf{F}_{211}+2\mathsf{F}_{222}-\mathsf{F}_{233}+\mathsf{F}_{323})  \\
  &  +\mathsf{a}_{21} (\mathsf{F}_{111}-\mathsf{F}_{133}+3 \mathsf{F}_{212}+2\mathsf{F}_{313})  
   +\mathsf{a}_{23} (2 \mathsf{F}_{113}+3\mathsf{F}_{223}-\mathsf{F}_{311}+\mathsf{F}_{333}) \\
  &  +2 \mathsf{a}_{12} \mathsf{F}_{212}+2 \mathsf{a}_{32} \mathsf{F}_{223}
  +\mathsf{a}_{13} (\mathsf{F}_{213}+\mathsf{F}_{312})+\mathsf{a}_{31} (\mathsf{F}_{123}+\mathsf{F}_{213})=0,
\end{aligned}
\end{equation} 
\begin{equation}
\begin{aligned}
& \mathsf{a}_{22}(\mathsf{F}_{111}-\mathsf{F}_{122}+\mathsf{F}_{313})
+\mathsf{a}_{33} (\mathsf{F}_{111}-\mathsf{F}_{133}+\mathsf{F}_{212}) 
+ \mathsf{a}_{11} (-2 \mathsf{F}_{111}+\mathsf{F}_{122}+\mathsf{F}_{133}-\mathsf{F}_{212}-\mathsf{F}_{313}) \\
&-\mathsf{a}_{12} (3\mathsf{F}_{112}+\mathsf{F}_{222}-\mathsf{F}_{233}+2 \mathsf{F}_{323}) 
 -\mathsf{a}_{13} (3 \mathsf{F}_{113}+2\mathsf{F}_{223}-\mathsf{F}_{322}+\mathsf{F}_{333}) \\
&-2 \mathsf{a}_{21} \mathsf{F}_{112}-2\mathsf{a}_{31} \mathsf{F}_{113}
-\mathsf{a}_{23} (\mathsf{F}_{123}+\mathsf{F}_{312})
-\mathsf{a}_{32} (\mathsf{F}_{123}+\mathsf{F}_{213}).
\end{aligned}
\end{equation} 
The above system of nonlinear PDEs are too complicated to solve analytically. However, we can analytically study cloaking an arbitrary cylindrical hole (see Fig. \ref{fig:cylindrical}).

\begin{figure}[hbt!]
\centering
\vskip 0.2in
\includegraphics[width=.5\textwidth]{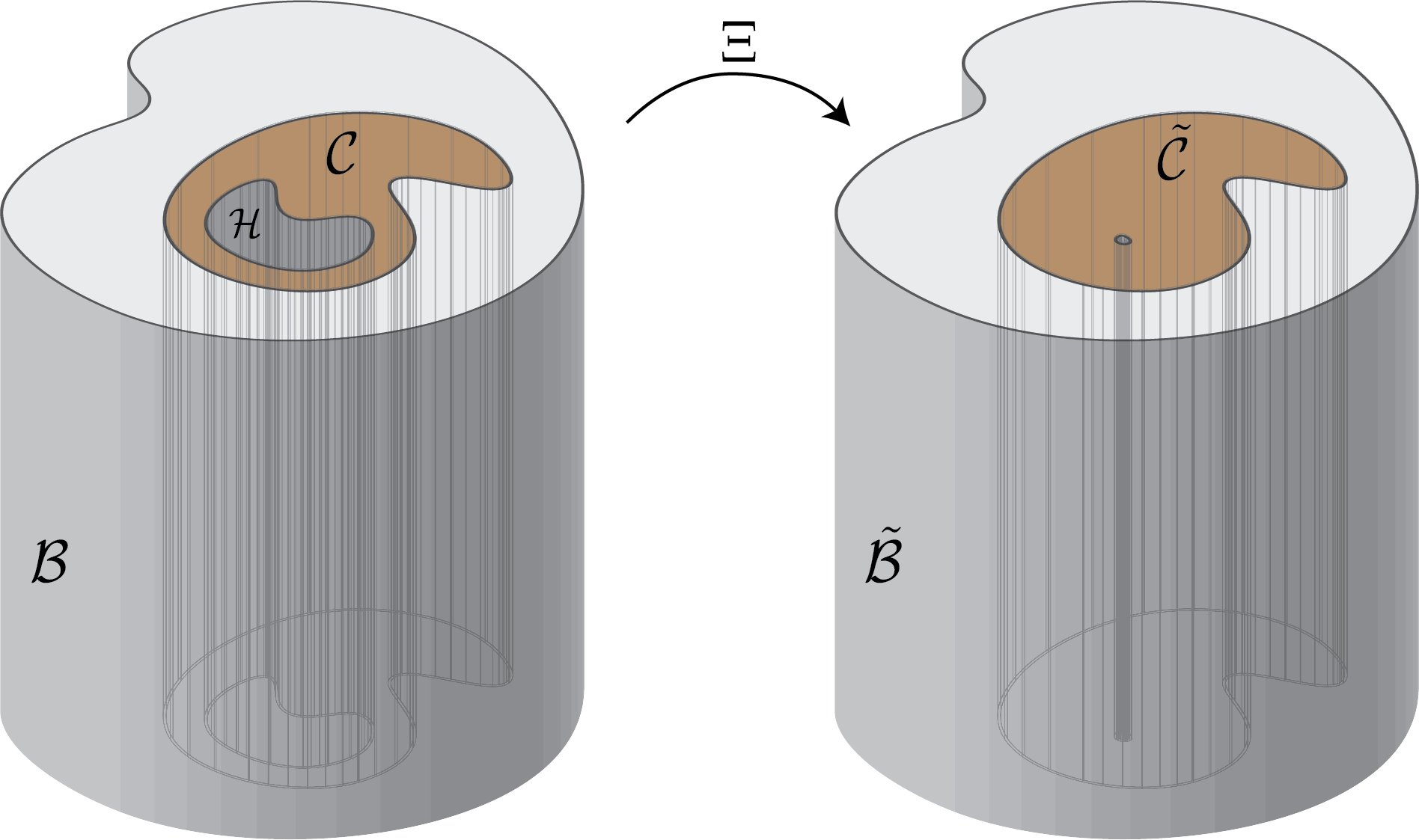}
\vskip 0.1in
\caption{A body with a cylindrical hole.}
\label{fig:cylindrical}
\end{figure}

\begin{prop}
\label{Cylinderical-Cloaking}
Assuming that the virtual body is isotropic and non-centrosymmetric, elastodynamic transformation cloaking is not possible for any cylindrical hole (not necessarily circular).
\end{prop}
\begin{proof}
Let us consider a cylindrical hole (not necessarily circular) that is covered by a cylindrical cloak. Let us assume that in the Cartesian coordinates $(X^1,X^2,X^3)$, the $X_3$ axis is the axis of the cylindrical hole. In this case the cloaking map has the form $\Xi(X^1,X^2,X^3)=(\tilde{X}^1,\tilde{X}^2,\tilde{X}^3)=(\Xi^1(X^1,X^2),\Xi^2(X^1,X^2),X^3)$. Therefore, 
$\cloak{\mathbf{F}}^{-1}$ and its covariant derivative have the following representations:
\begin{equation}
\cloak{\mathbf{F}}^{-1}=
\begin{bmatrix}
\mathsf{a}_{11} & \mathsf{a}_{12}  & 0  \\
\mathsf{a}_{21} & \mathsf{a}_{22}  & 0  \\
0 & 0  & 1  
\end{bmatrix},~~~\nabla\cloak{\mathbf{F}}^{-1}=
\renewcommand*{\arraystretch}{1.2}
\begin{bmatrix}
\begin{bmatrix}
\mathsf{F}_{111} & \mathsf{F}_{112}  & 0  \\
\mathsf{F}_{211} & \mathsf{F}_{212}  & 0  \\
0 & 0  & 0  
\end{bmatrix} \\
\begin{bmatrix}
\mathsf{F}_{112} & \mathsf{F}_{122}  & 0  \\
\mathsf{F}_{212} & \mathsf{F}_{222}  & 0  \\
0 & 0  & 0  
\end{bmatrix} \\
\begin{bmatrix}
0 & 0  & 0  \\
0 & 0  & 0  \\
0 & 0  & 0  
\end{bmatrix}
\end{bmatrix} .
\end{equation} 
Eq.\eqref{Equation1} is simplified to read
\begin{equation}
	(3 \lambda+2 \mu)(\mathsf{a}_{12}-\mathsf{a}_{21})^2+ \mu \left[3(\mathsf{a}_{11}-\mathsf{a}_{22})^2
	+3(\mathsf{a}_{12}+\mathsf{a}_{21})^2+(\mathsf{a}_{12}-\mathsf{a}_{21})^2\right]=0.
\end{equation} 
Knowing that $\mu>0$, and $3 \lambda+2 \mu>0$, one concludes that $\mathsf{a}_{12}=\mathsf{a}_{21}=0$, and $\mathsf{a}_{11}=\mathsf{a}_{22}$. Now, Eq.\eqref{Equation2} is simplified to read $3\mu(\mathsf{a}_{22}-1)^2=0$, which implies that $\mathsf{a}_{22}=1$. The other constraints are trivially satisfied. Therefore, $\Xi=\operatorname{id}$, which implies that cloaking is not possible.
\end{proof}

We suspect that transformation cloaking in dimension three is not possible for a cavity of any shape. 
\begin{conj}
Assuming that the virtual body is isotropic and non-centrosymmetric, elastodynamic transformation cloaking is not possible for a hole of any shape in dimension three.
\end{conj}

\section{Conclusions}

In this paper we investigated the possibility of transformation cloaking in non-centrosymmetric gradient solids. There have been claims in the literature that chirality can be utilized in achieving cloaking from stress waves. We formulated the transformation cloaking problem in terms of two equivalent boundary-value problems. We showed that transformation cloaking is not possible for any cylindrical hole (non necessarily circular). The obstruction to transformation cloaking is the balance of angular momentum. We were able to prove this negative result for holes with the topology of the $2$-sphere only for spheroidal holes and cloaking maps that preserve the spheroidal symmetry. We conjecture that exact transformation cloaking is not possible for a hole of any shape. Some of the existing works in the literature show approximate cloaking for some particular examples. They are, however, misleading as they are i) based on fundamentally flawed formulations that do not consider all the balance laws, and ii) one has no control over the errors. Our conclusion is that the path forward for engineering applications of elastodynamic cloaking is approximate cloaking formulated as an optimal design problem.

\section*{Acknowledgement}

This research was supported by ARO W911NF-18-1-0003 (Dr. Daniel P. Cole).

\bibliographystyle{abbrvnat}
\bibliography{ref}




\end{document}